\documentclass[10pt,journal,compsoc]{IEEEtran}

\ifCLASSOPTIONcompsoc
  \usepackage[nocompress]{cite}
\else
  \usepackage{cite}
\fi

\ifCLASSINFOpdf
  \else
 \fi

\usepackage{booktabs} % tables 
\RequirePackage[OT1]{fontenc}
\usepackage{amsthm,amsmath,amssymb,amsfonts,graphicx,subfigure,bm}
\usepackage{cite} 
\RequirePackage[colorlinks,citecolor=black,urlcolor=blue,linkcolor=black]{hyperref}
\RequirePackage{hypernat}
\usepackage{varioref}
\labelformat{section}{Section~#1}
\labelformat{figure}{Figure~#1}
\labelformat{table}{Table~#1}

\usepackage{cleveref}
\usepackage{autonum}

\usepackage{multirow}
\usepackage{rotating}
\usepackage[ruled]{algorithm2e}
\usepackage{fancybox}
\usepackage{tikz}
\usetikzlibrary{chains}
\usepackage{adjustbox}

\usepackage{array}
\newcolumntype{H}{>{\setbox0=\hbox\bgroup}c<{\egroup}@{}}

\usepackage{float}
\usepackage{balance}
\usepackage{nidanfloat}

\theoremstyle{plain}
\newtheorem{thm}{Theorem}%[section]

\newtheorem{prop}{Proposition}
\theoremstyle{remark}         % Also: \qedhere	!
\theoremstyle{definition} %[section]

\newcommand{\vect}[1]{\mbox{\boldmath $ #1$}}

\newcommand{\ind}{\stackrel{\mathrm{ind}}{\sim}}

\newcommand\om{\Omega}
\newcommand\real{\mathbb{R}}
\newcommand\A{\mathcal{A}}

\newcommand\T{\mathcal{T}}

\renewcommand\S{\mathcal{S}}

\newcommand\bs{\bm{s}}

\newcommand\bx{\bm{x}}
\newcommand\by{\bm{y}}
\newcommand\bz{\bm{z}}

\newcommand\bu{\bm{u}}

\newcommand\beps{\bm{\epsilon}}

\newcommand\bphi{\bm{\phi}}
\newcommand{\btau}{\bm{\tau}}

\newcommand\brho{\bm\rho}

\newcommand\blam{\bm\lambda}

\newcommand\R{\mathcal{R}}
\newcommand\D{\mathcal{D}}

\def\independenT#1#2{\mathrel{\setbox0\hbox{$#1#2$}%
		\copy0\kern-\wd0\mkern4mu\box0}}

\def\I{{\mathbf{1}}}

\newcommand{\papertitle}{Learning Asymmetric and Local Features in Multi-Dimensional Data through Wavelets with Recursive Partitioning}

\newcommand{\citep}{\cite} % natbib is not compatible with cite package 

\hypersetup{draft}
\begin{document} 

\title{\papertitle}

\author{Meng~Li
        and~Li~Ma% <-this % stops a space
\IEEEcompsocitemizethanks{\IEEEcompsocthanksitem M. Li is with the Department
of Statistics, Rice University, Houston, TX, 77025.
E-mail: meng@rice.edu. 
\IEEEcompsocthanksitem L. Ma is with the Department of Statistical Science, Duke University, Durham, NC, 27708. E-mail: li.ma@duke.edu.}% <-this % stops an unwanted space
%\thanks{Manuscript received XXXX, XXXX; revised XXXX, XXXX.}
}

\IEEEtitleabstractindextext{
\begin{abstract}
Effective learning of asymmetric and local features in images and other data observed on multi-dimensional grids is a challenging objective critical for a wide range of image processing applications involving biomedical and natural images. It requires methods that are sensitive to local details while fast enough to handle massive numbers of images of ever increasing sizes. We introduce a probabilistic model-based framework
that achieves these objectives by incorporating adaptivity into discrete wavelet transforms (DWT) through Bayesian hierarchical modeling, thereby allowing wavelet bases to adapt to the geometric structure of the data while maintaining the high computational scalability of wavelet methods---linear in the sample size (e.g., the resolution of an image). We derive a recursive representation of the Bayesian posterior model which leads to an exact message passing algorithm to complete learning and inference. While our framework is applicable to a range of problems including multi-dimensional signal processing, compression, and structural learning, we illustrate its work and evaluate its performance in the context of image reconstruction using real images from the ImageNet database, two widely used benchmark datasets, and a dataset from retinal optical coherence tomography and compare its performance to state-of-the-art methods based on basis transforms and deep learning.
\end{abstract}
}

\maketitle

\IEEEdisplaynontitleabstractindextext

\IEEEpeerreviewmaketitle

\IEEEraisesectionheading{\section{Introduction}\label{sec:intro}}

\IEEEPARstart{E}{ffective} learning of asymmetric and local features in images and other data observed on multi-dimensional grids plays a critical role in a wide range of applications. One such application is optical coherence tomography (OCT). OCT is a non-invasive imaging modality widely used in ophthalmology to visualize cross-sections of tissue layers. These tissue layers---such as the inner nuclear layer and outer nuclear layer---are often mostly homogeneous horizontally while involving large vertical contrasts. These contrasts across layers are key for ophthalmologists to make a diagnosis based on the (algorithmically reconstructed) image. Furthermore, local structures in such images can indicate ocular diseases, and their proper quantitative assessment is an important reference for monitoring the progression of the disease in clinical practice~\citep{alasil2010relationship,bussel2013oct,huang2014inner,sun2014disorganization,kafieh2015thickness,oishi2018longitudinal}. Many other applications of 2D and 3D image analyses in biomedicine and beyond also involve asymmetric and local features to various extents. The effective analysis of such multi-dimensional observations can be greatly enhanced by incorporating adaptivity into the algorithm or method to take into account such features.

A further challenge in modern applications involving multi-dimensional observations is the ever increasing size of the datasets. For example, both the number of images analyzed as well as the resolution---i.e., the total number of pixels---of each image have been expanding rapidly. Many traditional methods and models become computational impractical for modern data as they scale polynomially with the resolution. Effective methods for analyzing such data must scale well with both the resolution of each image and the number of images.

The primary aim of this work is to present a general-purpose generative probabilistic model for data on multi-dimensional grids that can be used to address these challenges in inference and learning---being able to effectively adapt to the asymmetric and local nature of interesting features, while achieving a highly efficient linear computational budget. 

Our starting point is a well-known strategy for representing functions---a multi-resolution representation through the discrete wavelet transform (DWT). Wavelet analysis is hardly a new topic~\citep{Donoho1994,donoho1995adapting,Mallat2008} and it has played an important role in the context of signal processing and image analysis. Its linear computational scalability is well-suited for analyzing massive data. However, traditional statistical wavelet analyses have mostly been focusing on effective modeling and inference on the wavelet coefficients {\em given} a {\em fixed} wavelet transform of the original data~\citep{Abramovich+:98,crouse:1998,Clyde+George:00,Brown2001,Wil+Now:04,Morris2006}. A predetermined fixed wavelet transform, however, cannot adapt to the structure of the data and consequently suffers in its ability to effectively maintain the local structures in the original observation. Also, classical wavelet transforms on multi-dimensional grids are generally symmetric with respect to the dimensions, rendering them ineffective for preserving asymmetric features. No downstream statistical analyses can recover what has already been lost at the upstream wavelet transform stage. 

In this work, we show that it is possible to incorporate the desired adaptivity into the wavelet transform stage while maintaining the computational scalability of the statistical analysis through a very simple hierarchical modeling strategy---starting the model ``one level up'', that is, by incorporating the wavelet transform itself as an unknown quantity of interest into the probabilistic model, and learn it based on the data. Specifically, 
we consider latent (1D) wavelet transforms that can ``twist and turn'' (or ``warp'') over the multi-dimensional grid, or the {\em index space}, and adopt a Bayesian prior on the path of its twisting and turning. In other words, we place a prior on the local directionality of the 1D transform to allow the ``warping'' to adapt to the geometric structure of the underlying function, e.g., the true image, through the Bayesian machinery.

In designing an appropriate prior for the local directionality, we note that ``warping'' a 1D wavelet transform through the grid points is equivalent to fixing the 1D wavelet transform while shuffling grid points in the multivariate index space of the observation---i.e., through applying a given 1D wavelet transform to a permuted version of the observation. This connection implies that probabilistic models on ``warping'' can be induced from distributions on the space of permutations of the index points or locations. Moreover, we draw a further connection between permutations and recursive dyadic partitioning on the index space to construct a prior on permutations induced by random recursive partitioning over the index space. This prior takes advantage of the fact that multi-dimensional images tend to be piecewise smooth to strike a balance between flexibility and computational tractability, allowing us to complete exact Bayesian inference through a recursive message passing algorithm with a computational budget linear in the resolution and sample size. 

Due to the connection to recursive partitioning, we shall refer to our approach as WARP, or {\em WAvelets with Recursive Partitioning}. Through extensive numerical studies involving a large number of natural images from the ImageNet database, two additional benchmark data sets, and an OCT data set, we show that WARP often outperforms the existing state-of-the-art approaches by a substantial margin while maintaining the computational efficiency of classical wavelet analyses with fixed wavelet transforms. While we focus on 2D and 3D images in our motivation and numerical examples, our framework is readily applicable to observations of more than three dimensions without modification.

The rest of the paper is organized as follows. \ref{section:method} introduces the WARP framework. In Section~\ref{section:RDP} we review the key components of Bayesian wavelet regression models, introduce permutation of the index space as a way to incorporate adaptivity into wavelet analysis, and construct a class of priors on permutations induced by recursive dyadic partitioning on the index space. We derive the corresponding posterior model and provide computational recipes for exact Bayesian inference under the WARP model with Haar wavelets in Section~\ref{sec:posterior.inference}. In~\ref{sec:experiments}, we carry out an extensive numerical study and compare our method to existing state-of-the-art wavelet and non-wavelet methods including a deep learning method on a variety of real images. In \ref{sec:application} we carry out a case study by applying WARP to analyze an OCT data set, and compares its performance to a number of state-of-the-art approaches. \ref{sec:dicussion} concludes with some brief remarks. 
The C++ source code along with a Matlab toolbox and R package to implement the proposed method is available online at \url{https://github.com/MaStatLab/WARP}. 

\section{Method\label{section:method}} 
\subsection{A Bayesian hierarchical wavelet regression model with recursive dyadic partitions\label{section:RDP}}
\subsubsection{Background and overview}
We use $\om$ to denote a space of indices or locations (e.g., pixels in images) where we obtain numerical measurements (e.g., intensities of pixels). Throughout this work, we assume $\om$ to be an $m$-dimensional rectangular tube consisting of $n_i=2^{J_i}$ grid points in the $i$th dimension for $i=1,2,\ldots,m$, that is, the function values are observed on a multi-dimensional equidistant grid. To simplify notation, we shall use $[a,b]$ to represent the set $\{a,a+1,\ldots,b\}$ for two integers $a$ and $b$ with $a\leq b$. Then the index space $\om$ is of the form
\[
\om=[0,2^{J_1} - 1]\times [0,2^{J_2} - 1]\times \cdots \times [0,2^{J_m} - 1].
\]
The locations in $\om$ can be placed into a vector of length $n=2^J$. 
For example, we can map the location $\bs=(s_1,s_2,\ldots,s_m) \in \om$ to the $t$th element in the vector, where $t = s_1 + \sum_{l = 2}^m  (\prod_{i = 1}^{l -1} n_i ) s_{l}$. Correspondingly, any function $f:\om \rightarrow \real$ can be represented as a vector ${\bm f}$ of length $n=2^{J}$ whose $t$th element is $f(\bs)$.

Now, we consider a regression model
\begin{equation}\label{eq:model}
\by = \bm{f} + \beps \quad \text{with } \beps \sim N(\bm{0}, \Sigma_{\epsilon}),
\end{equation}
where $\by=(y_0,y_1,\ldots,y_{2^{J}-1})'$ are the observations on $\om$, $\bm{f}=(f_0,f_1,\ldots,f_{2^{J}-1})'$ the underlying unknown function mean (or the signal), and $\beps=(\epsilon_0,\epsilon_1,\ldots,\epsilon_{2^{J}-1})'$ the noise. For ease of illustration, we assume homogeneous white noise, i.e., $\Sigma_{\epsilon} = \sigma^2 I_n$, though our model and inference algorithms do not rely on this assumption at all and can be readily apply to models with heterogeneous variance; see~\ref{sec:dicussion} for further discussion.

One can apply a 1D discrete wavelet transform (DWT) to the observation vector $\by$ through multiplying the corresponding orthonormal matrix $W$ to both sides of Eq.~\eqref{eq:model}, obtaining $\bm{w} = \bz + \bu$ where $\bm{w} = W \by$ is the vector of empirical wavelet coefficients, $\bz =W \bm{f}$ the mean vector for wavelet coefficients and $\bu =W\beps$ the noise vector in the wavelet domain. This model can be rewritten in a location-scale form: $w_{j,k} = z_{j,k}+ u_{j,k}$ for $j=0,1,\ldots,J-1$ and $k=0,1,\ldots,2^{j}-1$,
where $w_{j,k}$, $z_{j,k}$, $u_{j,k}$ are the $k$th wavelet coefficient, signal, and noise at the $j$th scale in the wavelet (i.e., location-scale) domain, respectively. 

It will generally be unreasonable to treat multi-dimensional observations simply as a vector with an arbitrary ordering of the locations; 
see~\cite{Donoho:99,Jac+:11,Ali+:13}. Such a vectorization
ignores the structure of the underlying function, and thus will result in less effective ``energy concentration'', i.e., producing 
a wavelet decomposition of ${\bm f}$ that is not very sparse---with many non-zero $z_{j,k}$'s of small to moderate sizes, reducing the signal-to-noise ratio at those $(j,k)$ combinations. 

For each specific data set at hand, however, there typically exists some orderings of the locations that effectively reorganize the data so that the corresponding vectorization of the data provides an efficient
representation of the underlying function; see~\ref{fig:RDP} for an illustration. Adopting a Bayesian modeling perspective, one can think of the underlying ``good'' vectorizations as latent structures of interest. Also, one can view the wavelet regression model under each given index permutation as a competing generative model for the observed data given the latent structure. 
This perspective 
inspires us to incorporate a prior on the permutations, thereby allowing us to compute a posterior on the space of competing wavelet regression models, and then based on the goal of the analysis proceed with the common devices for Bayesian inference. Two particular useful tools are (i) Bayesian model selection~\citep{raftery1995bayesian}---learning a good permutation for representing the image based on its posterior probability; and (ii) Bayesian model averaging---estimating the underlying function based on averaging over the different permutations using their posterior probabilities~\citep{Volinsky1999}.

This Bayesian approach does incur a practical challenge commonly arising in high-dimensional problems---the space of all permutations is so massive that brute-force enumeration over the space is computationally prohibitive. In the current context, effective exploration of the model (i.e., permutation) space becomes possible, however, once we realize that the vast majority of the permutations will lead to wavelet regression models that ignore the spatial smoothness of the underlying function---i.e., close locations in $\om$ often correspond to similar values in ${\bm f}$. In particular, we can focus attention on a subclass of permutations that to various extents preserve spatial smoothness, and design a model space prior supported on this  manageable subclass. To this end, we appeal to a relationship between recursive dyadic partitioning (RDP) and permutations, and shall consider the collection of permutations induced by RDPs on $\om$. 

Next we introduce some basic notions regarding RDPs on $\om$, which are then used to construct a prior on permutations. In reading the next two subsections, the reader may refer to \ref{fig:RDP} for an illustration of the key notions and notations.

\subsubsection{Recursive dyadic partitioning on the location space} 

A {\em partition} of $\om$ is a collection of nonempty sets $\{A_1,A_2,\ldots,A_H\}$ such that $\om=\cup_{h=1}^{H} A_h$ and $A_{h_1}\cap A_{h_2}=\emptyset$ for any $h_1\neq h_2$. Now let $\T^{0},\T^1,\T^{2},\ldots,\T^{j},\ldots$ be a {\em sequence of partitions} of $\om$. We say that this sequence is a {\em recursive dyadic partition} (RDP) if it satisfies the following two conditions: (i) $\T^{j}$ consists of $2^{j}$ blocks: $\T^{j}=\{A_{j,k}:k=0,1,\ldots,2^{j}-1\}$; (ii) $\T^{j+1}$ is obtained by dividing each set in $\T^{j}$ into two pieces, i.e., $A_{j,k} = A_{j+1,2k}\cup A_{j+1,2k+1}$ for all $j\geq 0$ and $k=0,1,\ldots,2^{j}-1$.

We call an RDP {\em canonical} if the sequence of partitions satisfy two additional conditions: (iii) if the partition blocks $A_{j,k}$ are rectangles of the form
\[
A_{j,k} = [a_{j,k}^{(1)},b_{j,k}^{(1)}] \times [a_{j,k}^{(2)},b_{j,k}^{(2)}] \times \cdots \times [a_{j,k}^{(m)},b_{j,k}^{(m)}].
\]
and (iv) $A_{j+1,2k}$ and $A_{j+1,2k+1}$ are produced by dividing $A_{j,k}$ into two halves at the middle of one of $A_{j,k}$'s \textit{divisible} dimensions.

A rectangular partition block $A_{j,k}$ is {\em divisible} in dimension $d$ if $A_{j,k}$ is supported on at least two values in that dimension, i.e., $a_{j,k}^{(d)} < b_{j,k}^{(d)}$. In this case, if $A_{j,k}$ is divided in dimension~$d$, then its children $A_{j+1,2k}$ and $A_{j+1,2k+1}$ are given by
\begin{equation}
[a_{j+1,2k}^{(d)},b_{j+1,2k}^{(d)}] = [a_{j,k}^{(d)},(a_{j,k}^{(d)}+b_{j,k}^{(d)})/2] 
\end{equation} and 
\begin{equation}
[a_{j+1,2k+1}^{(d)},b_{j+1,2k+1}^{(d)}] = [(a_{j,k}^{(d)}+b_{j,k}^{(d)})/2+1,b_{j,k}^{(d)}], 
\end{equation}
while
\[
[a_{j+1,2k}^{(d')},b_{j+1,2k}^{(d')}] = [a_{j+1,2k+1}^{(d')},b_{j+1,2k+1}^{(d')}] = [a_{j,k}^{(d')},b_{j,k}^{(d')}] 
\]
for all $d'\neq d$.

Any canonical RDP on $\om$ will have exactly $J+1$ levels, i.e., $\T^{0},\T^{1},\ldots,\T^{J}$. The $j$th level partition $\T^{j}$ consists of $2^{j}$ rectangular pieces of equal size, each covering $n/2^{j}$ locations in $\om$. From now on, we simply use RDP to refer to canonical ones when this causes no confusion.

\subsubsection{RDPs and permutations}
Each RDP can be represented by a $J$ level bifurcating tree with the partition blocks in $\T^{j}$ forming the $2^j$ nodes in the $j$th level of the tree. As such, we can use $\T=\cup_{j=0}^{J} \T^{j}$ to represent the RDP. Each node in the $J$th level corresponds to a unique location in $\om$, and is called ``atomic" as it contains a single element. We shall interchangeably refer to an RDP as a ``tree'', and to the partition blocks as ``nodes''.

Given the RDP $\T$, each location $\bs \in \om$ falls into a unique branch of $\T$, that is, $\om=A_{0,0}\supset A_{1,k_1(\bs)}\supset A_{2,k_2(\bs)}\supset \cdots \supset A_{J,k_J(\bs)}=\{\bs\}$, with $A_{j,k_j(\bs)}$ being the node in the $j$th level to which $\bs$ belongs. Accordingly, the RDP $\T$ induces a unique vectorization of the locations in $\om$ such that $\bs$ corresponds to the $t(\bs)$th element of the vector where $t(\bs)=\sum_{l = 1}^J  2^{J-l}\cdot  e_l(\bs)$ with $e_l = k_l(\bs)\!\! \mod 2$, indicating the branch of the tree $\bs$ falls into at level $l$. As such, $\T$ induces a permutation of the $n$ locations, and we let $\pi_{\T}$ denote this permutation. 

As an illustration, \ref{fig:RDP} presents an RDP and the induced permutation using a toy $4 \times 4$ image (so $m = 2$ and $J_1 = J_2 = 2$). We index pixels in the true image from 0 to 15. In addition, we assume that the underlying function takes only two values---1 and 2---on the 16 locations, represented by the white and the red colors, respectively. The demonstrated RDP corresponds well to the structure of the underlying signal, which would result in an effective 1D wavelet analysis on the vectorized observation.

\def\pixels{
	{2,2,0,2},
	{2,2,0,2},
	{0,0,0,2},
	{2,2,0,2}%
}

\def\pixelsvec{
	{2,2,2,2,0,0,2,2,0,0,0,0,2,2,2,2}%
}

\definecolor{pixel 0}{HTML}{FFFFFF}
\definecolor{pixel 2}{HTML}{FF0000}

\tikzstyle{level 1}=[level distance=3cm, sibling distance=4cm]
\tikzstyle{level 2}=[level distance=3.5cm, sibling distance=2cm]
\tikzstyle{level 3}=[level distance=3.5cm, sibling distance=1cm]
\tikzstyle{level 4}=[level distance=3.5cm, sibling distance=0.5cm]
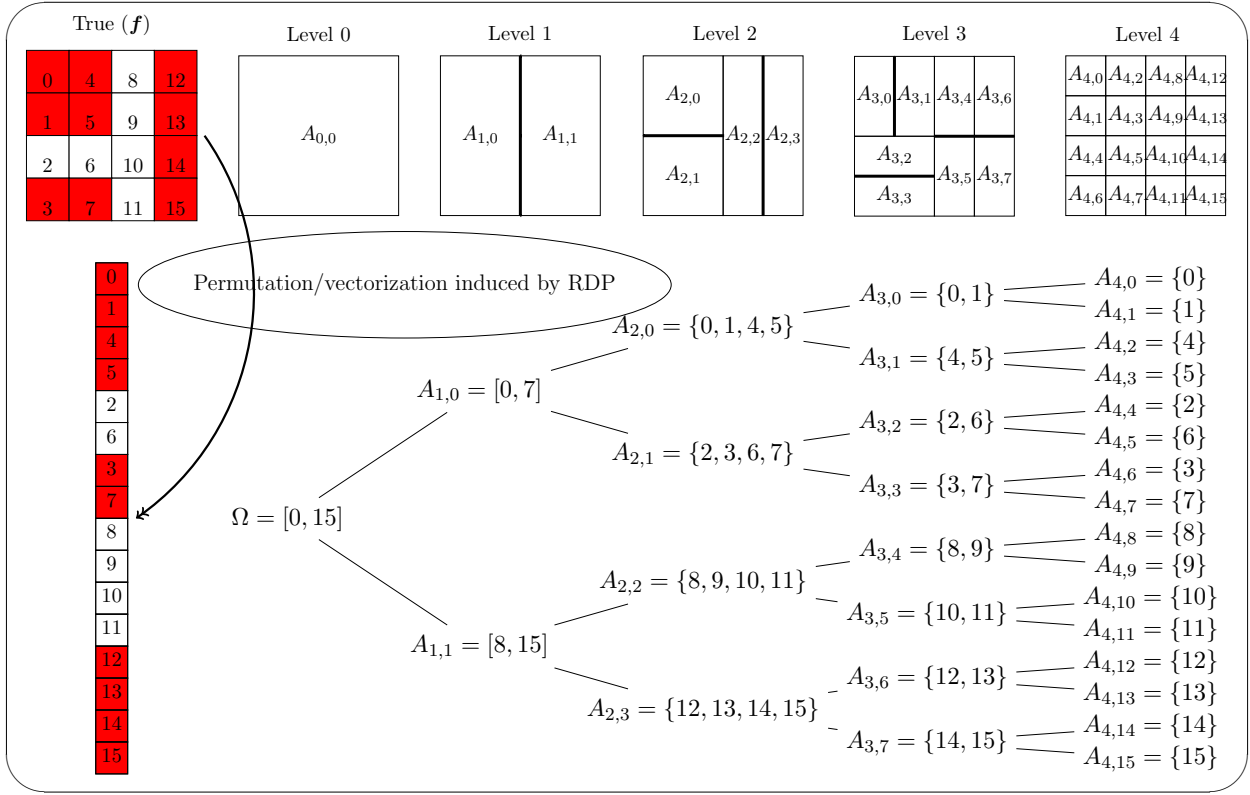
\begin{figure*}[ht!]
	\centering
	\begin{adjustbox}{width=0.8\textwidth,height=0.9\textheight,keepaspectratio}
		\ovalbox{
			\begin{tikzpicture}
			\begin{scope}[start chain=1 going right, node distance=5mm]
			\node[on chain=1, label={True ($\vect{f}$)}] (Obs) {
				\begin{tikzpicture}[scale = 0.8]
				\foreach \line [count=\y] in \pixels {
					\foreach \pix [count=\x] in \line {
						\pgfmathsetmacro{\z}{int((\x -1)*4 + \y-1)};
						\draw[fill=pixel \pix] (\x,-\y) rectangle +(1,1) node at +(0.5, 0.3) {\z};
				}}\end{tikzpicture}};
			\node [on chain=1, label={Level 0}] (A) {\begin{tikzpicture}[anchor=center, scale=1.5]
				\draw (0, -2) rectangle (2,0) node[pos=.5] {$A_{0, 0}$}; 
				\end{tikzpicture}};
			\node[on chain=1, label={Level 1}] (L1) {
				\begin{tikzpicture}[anchor=center, scale=1.5]
				\draw (0,-2) rectangle (1,0) node[pos=.5] {$A_{1, 0}$}; 
				\draw (1,-2) rectangle (2,0) node[pos=.5] {$A_{1, 1}$};
				\draw[black, ultra thick] (1,-2) grid (1,0);
				\end{tikzpicture}
			};
			\node[on chain=1, label={Level 2}] (L2) {\begin{tikzpicture}[anchor=center, scale=1.5]
				\draw (0,-1) rectangle (1,0) node[pos=0.5] {$A_{2,0}$};
				\draw (0, -2) rectangle (1,-1) node[pos=0.5] {$A_{2,1}$};
				\draw (1, -2) rectangle(1.5, 0) node[pos=.5]{$A_{2, 2}$}; 
				\draw (1.5, -2) rectangle(2, 0) node[pos=.5]{$A_{2, 3}$}; 
				\draw[black, ultra thick] (0,-1) grid (1,-1);
				\draw[black, ultra thick] (1.5,-2) -- (1.5,0);
				\end{tikzpicture}};
			\node[on chain=1, label={Level 3}] (L3) {\begin{tikzpicture}[anchor=center, scale=1.5]
				\draw (0,-1) rectangle (0.5, 0) node[pos=0.5] {$A_{3,0}$};
				\draw (0.5,-1) rectangle (1, 0) node[pos=0.5] {$A_{3,1}$}; 
				\draw (0, -1.5) rectangle (1, -1) node[pos=0.5] {$A_{3,2}$}; 
				\draw (0, -2) rectangle (1, -1.5) node[pos=0.5] {$A_{3,3}$};
				\draw (1, -2) rectangle (1.5, -1) node[pos=0.5] {$A_{3,5}$}; 
				\draw (1, -1) rectangle +(0.5, 1) node[pos=0.5] {$A_{3,4}$}; 
				\draw (1.5, -2) rectangle +(0.5, 1) node[pos=0.5] {$A_{3,7}$};
				\draw (1.5, -1) rectangle +(0.5, 1) node[pos=0.5] {$A_{3,6}$};
				\draw[black, ultra thick] (0.5,-1) -- (0.5,0); 
				\draw[black, ultra thick] (0, -1.5) -- (1, -1.5);
				\draw[black, ultra thick] (1, -1) -- (1.5, -1); 
				\draw[black, ultra thick] (1.5, -1) -- (2, -1); 
				\end{tikzpicture}};
			\node[on chain=1, label={Level 4}](L4){
				\begin{tikzpicture}[anchor=center,scale = 0.75]
				\foreach \i [count=\j from 0] in {0,1,4,5,2,6,3,7,8,9,10,11,12,13,14,15}{
					\pgfmathsetmacro{\x}{int(mod(\i, 4)};
					\pgfmathsetmacro{\y}{int((\i - \x) / 4)};
					\draw (\y,-\x) rectangle +(1,-1) node[pos=.5] {$A_{4, \j}$};
				}			
				\end{tikzpicture}}; 
		\end{scope}
		\begin{scope}[start branch=B2]
			\node[below=0.5cm of Obs] (C) {
				\begin{tikzpicture}[scale=0.6]
				\foreach \line [count=\y] in \pixelsvec {
					\foreach \pix [count=\x] in \line {
						\pgfmathsetmacro{\z}{int((\x -1))};
						\draw[fill=pixel \pix] (\y,-\x) rectangle +(1,-1); %node at +(0.5, 0.01) {\z};
					}
				};			
				\foreach \i [count=\j] in {0,1,4,5,2,6,3,7,8,9,10,11,12,13,14,15}
				\draw (1,-\j) rectangle +(1,-1) node at +(0.5, 0.01) {\i};
				\end{tikzpicture}
			}; 
			\node[below=0.5cm of L2]  (B) {
				\begin{tikzpicture}[anchor=center, grow=right, yscale=-1, draw=none, scale=1.2, every node/.style={scale=1.2}]
				\node[below=of A] (z){$\Omega=[0,15]$}
				child {
					node {$A_{1, 0}=[0,7]$}
					child {node [draw=none] {$A_{2,0}=\{0,1,4,5\}$ }
						child {node {$A_{3,0}=\{0,1\}$}
							child {node {$A_{4, 0}=\{0\}$}} child {node {$A_{4, 1}=\{1\}$}}
						} 
						child {node {$A_{3,1}=\{4,5\}$}
							child {node {$A_{4, 2}=\{4\}$}} child {node {$A_{4, 3}=\{5\}$}}
					}}
					child {node [yshift=0pt] {$A_{2,1}=\{2,3,6,7\}$}
						child {node {$A_{3,2}=\{2,6\}$}
							child {node {$A_{4, 4}=\{2\}$}}
							child {node {$A_{4, 5}=\{6\}$}}
						} 
						child {node {$A_{3,3}=\{3,7\}$}
							child {node {$A_{4, 6}=\{3\}$}}
							child {node {$A_{4, 7}=\{7\}$}}}
					}
				}
				child {node {$A_{1, 1}=[8, 15]$}
					child {node {$A_{2,2}=\{8, 9, 10, 11\}$ }
						child {node {$A_{3,4}=\{8, 9\}$}
							child {node {$A_{4, 8}=\{8\}$}} child {node {$A_{4, 9}=\{9\}$}}
						} 
						child {node {$A_{3,5}=\{10,11\}$}
							child {node {$A_{4, 10}=\{10\}$}} child {node {$A_{4, 11}=\{11\}$}}
					}}
					child {node {$A_{2,3}=\{12,13,14,15\}$}
						child {node {$A_{3,6}=\{12,13\}$}
							child {node {$A_{4, 12}=\{12\}$}}
							child {node {$A_{4, 13}=\{13\}$}}
						} 
						child {node {$A_{3,7}=\{14,15\}$}
							child {node {$A_{4, 14}=\{14\}$}}
							child {node {$A_{4, 15}=\{15\}$}}}
					}
				}; 
				\end{tikzpicture}}; 
		\end{scope}	
		\begin{scope}
			\draw[->>,  very thick] (Obs.east) [out=-55, in=35] to (C.east); 
			\draw (5.5, -2.8) ellipse (5cm and 1cm) node {\large Permutation/vectorization induced by RDP};
		\end{scope}
\end{tikzpicture}}
\end{adjustbox}
\caption{Illustration of the correspondence between RDPs and permutations. In the tree representation, $A_{2, 0} = \{0, 1, 4, 5\}$ means the node $A_{2, 0}$ contains the (0, 1, 4, 5)th elements of $\Omega$. The coloring code for the observations is red for 2 and white for 1. From level 0 to level 3, edges that are thicker than others are the partitions of the current level; nodes at the last level are all atomic.}
\label{fig:RDP}
\end{figure*}

We shall now utilize the relationship between RDPs and permutations to construct a prior on the latter. Before that, we shall simplify our notations a little. Note that while what the $(j, k)$th node $A_{j, k}$ is depends on the RDP $\T$, different RDPs can share common nodes---the $(j,k)$th node in one $\T$ may be the same as the $(j,k')$th node in another. (Note that the level of the node must be the same in either RDP.) In the following, we will need to specify quantities that only depend on the node regardless of the RDP tree $\T$ it arises from. A succinct way for expressing such quantities is to write them as a mapping from $\A$ to $\real$, where $\A$ denotes the collection of all sets that {\em could} be nodes in {\em some} RDP, or equivalently, $\A$ is the totality of nodes in all RDPs.
(This is to be distinguished from the collection of nodes in any particular RDP, which is denoted by $\T$.) It is worth noting $\A$ is a finite set.

Now we may define $\rho_{j,k}$ in such a way that its value only depends on what the set $A_{j,k}$ is, regardless of the RDP $\T$ to which it belongs. In this case we can let $\rho_{j,k} = \rho(A_{j, k})$, where $\rho(\cdot)$ is a mapping form $\A$ to $[0,1]$. While a set $A\in \A$ might be the $A_{j,k}$ in one RDP and $A_{j,k'}$ in another, the corresponding $\rho(A)$ value will then be the same under this mapping based specification.
The mapping-based notation such as $\rho(\cdot)$ allows various parameters to be specified in a node-specific (rather than RDP-specific) manner. This observation has extremely important computational implications---as we will show later, 
the space of nodes $\A$ for all canonical RDPs is of a cardinality linear in the size of $\om$, while that of canonical RDPs is exponential in $n$. (See Proposition~\ref{lem:complexity.tree} in the Supplementary Materials.) Therefore it is exactly the ability to carrying out the computation for the posterior in a node-specific manner that allows us to achieve linear complexity in our inference algorithm. Moreover, this notation will also help elucidate derivations on the posterior.

\subsubsection{Priors on RDPs: random RDP\label{section:priors.RDP}}
Our strategy of representing multi-dimensional functions using vectors will only pay off if the vectorization of $\om$ can result in an efficient characterization of the data, thereby leading to stronger energy concentration under wavelet transforms. For example, the RDP illustrated in \ref{fig:RDP} will lead to particularly efficient inference of the corresponding function. In general, the true optimal vectorization---or the corresponding RDP---is unknown, and one shall rely on the data to learn the RDPs that induce ``good'' vectorizations. 

We aim to achieve this in a hierarchical Bayesian approach by treating the RDP as a latent structure and placing a prior on the RDP. We consider the following prior on the RDP originally proposed in the context of density estimation~\citep{wongandma:2010,ma:2013}, which is specified in a node-specific fashion and leads to very efficient node-based posterior inference algorithms that scale linearly in $n$, the size of $\om$.

We describe the prior as a simple generative procedure for an RDP in an inductive manner. First, $\T^{0}=\{\om\}$ by definition. Now suppose we have generated $\T^{0},\T^{1},\cdots,\T^{j}$ for some $0\leq j \leq J-1$, then $\T^{j+1}$ is generated as follows. For each $A_{j,k}\in\T^{j}$, let $\D(A_{j,k})\subset \{1,2,\ldots,m\}$ be the collection of its divisible dimensions. We randomly draw a dimension in $\D(A_{j,k})$, and divide $A_{j,k}$ in that dimension to get $A_{j+1,2k}$ and $A_{j+1,2k+1}$. In particular, we let $\lambda_{d}(A_{j,k})$ be the probability for drawing the $d$th dimension,
where $\sum_{d=1}^{m} \lambda_{d}(A_{j,k}) = 1$ and $\lambda_{d}(A_{j,k})=0$ for $d\not\in \D(A_{j,k})$. In many problems, {\em a priori} one has no reason to favor dividing any particular dimension over another, and a default specification is to set
\[ \lambda_{d}(A_{j,k})=1/|\D(A_{j,k})|\cdot \I_{\{d\in\D(A_{j,k})\}},\]
where $\I_{E}$ is the indicator function of whether $E$ holds or not. This completes the inductive generation of $\T^{j+1}$. The procedure will terminate after $\T^{J}$ is generated as all nodes in $\T^{J}$ are atomic with no divisible dimensions.

The above generative mechanism forms a probability distribution on the space of RDPs, which is called the {\em random recursive dyadic partition} (RRDP) distribution, and it is specified by the collection of selection probabilities $\lambda_{d}(\cdot)$ defined on all {\em potential} nodes. We write
\begin{equation}
\T \sim \text{RRDP}(\blam),
\end{equation}where
$\{ \blam(A): A\in\A\},$ and  $\blam(A)=(\lambda_1(A),\lambda_2(A),\ldots,\lambda_{m}(A))'$, that is, $\blam$ is a mapping from $\A$ to the $(m-1)$-dimensional simplex. 

It is worth noting that the RRDP is a restricted version of the more general Bayesian classification and regression tree (CART) prior \citep{Chipman1998a,Denison1998}. The main constraint in RRDP compared to the general Bayesian CART is that the former is supported on canonical RDPs only---that is, each dyadic partition must be an even split, occurring at the middle of the range in one of the divisible dimensions. This additional restriction ensures that the cardinality of $\A$ is linear in $n$, thereby 
reducing the computational complexity required for inference to $O(n)$.

\subsection{Recipes for Bayesian inference\label{sec:posterior.inference}}

In this section, we present recipes for deriving and sampling from the posterior of our Bayesian model, and for evaluating posterior summaries such as the posterior mean of ${\bm f}$. We note that the marginal posterior of the RDP $\T$ is the key component for posterior inference, because once conditional on $\T$, our model reduces to a standard Bayesian wavelet regression for which closed-form conditional posteriors are readily available under common prior specifications.

Interestingly, when a Haar basis is adopted in the wavelet regression model, the marginal posterior of $\T$ can be calculated analytically in closed form through a recursive algorithm that is operationally similar to Mallat's pyramid algorithm, achieving a linear computational complexity $O(n)$.

\subsubsection{Exact Bayesian inference under Haar basis\label{section:estimation}}
The Haar wavelet basis is unique in that 
the $(j,k)$th wavelet coefficient under the vectorization induced by any RDP $\T$ 
is determined by only the locations inside the node $A_{j,k}$. We call this property of the Haar basis {\em node-autonomy} and say that inference under the Haar basis is {\em node-autonomous}. Specifically, for all RDPs in which $A$ is a node and is divided in the $d$th direction, the corresponding Haar wavelet coefficient associated with the node $A$ is given by 
\[
w_{d}(A) = 1/\sqrt{|A|} \cdot \left(\sum_{\bx\in A^{(d)}_{l}}y(\bx)-\sum_{\bx\in A^{(d)}_{r}}y(\bx)\right)
\] 
where $A^{(d)}_{l}$ and $A^{(d)}_{r}$ represent the two children nodes if $A$ is divided in the $d$th dimension and $|A|=2^{J-j}$ is the total number of locations in $A$. In contrast, wavelet coefficients from wavelet bases with longer support than Haar are not node-autonomous---not only does the coefficient associated with $A$ depend on the observations within $A$ but on those in other (often but not always adjacent) nodes in $\T$ as well. 

Node-autonomy enables the posterior to be computed in a node-specific fashion, 
avoiding integration in the much larger space of RDPs. Consequently, exact inference can be completed in a computational complexity of the same scale as the total number of all potential nodes in RDPs, which is equal to $\prod_{i = 1}^m (2 n_i - 1) = O(2^m n)$. 

Next we lay out the general strategy for inference. 
We show through two theorems that the marginal posterior of the RDP $\T$  is computable in analytically through a recursive algorithm that resembles Mallat's pyramid algorithm for two very popular classes of Bayes wavelet regression models---(i) those that model each wavelet coefficient independently given $\T$ (Theorem~\ref{thm:independent_shrinkage}); and (ii) those that induce a hidden Markov model (HMM) for incorporating dependency among the wavelet coefficients given $\T$ (Theorem~\ref{thm:latent_state}). 

\begin{thm}
	\label{thm:independent_shrinkage}
	Suppose $\T \sim {\rm RRDP}(\blam)$ and given the Haar DWT under $\T$, one models the wavelet coefficients independently, i.e., $(w_{j,k},z_{j,k}) \ind p_{j,k}(w,z\,|\,\bphi)$ for all $(j,k)$, where $\bphi$ represents the hyperparameters of the Bayesian wavelet regression model.
	Then the marginal posterior of $\T$ is still an RRDP. Specifically,
	$\T\,|\,\by \sim {\rm RRDP}(\tilde{\blam})$
	where the posterior selection probability mapping $\tilde{\blam}$ is given as
	\[
	\tilde{\lambda}_{d}(A) =\lambda_{d}(A) M_{d}(A) \Phi(A^{(d)}_{l})\Phi(A^{(d)}_{r})/\Phi(A)
	\]
	for any non-atomic $A\in \A$ where $M_{d}(A)$ is the marginal likelihood contribution from the wavelet coefficient on node $A$ if it is a node in $\T$ and divided in dimension $d$, i.e.,
	$M_{d}(A) = \int p_{j,k}(w_{d}(A),z\,|\,\bphi) \, dz$
	and $\Phi:\A \rightarrow [0,\infty)$ is a mapping defined recursively (i.e., its value on $A$ depends on its values on $A$'s children) as
	\begin{equation}
	\Phi(A) = \sum_{d\in \D(A)} \lambda_{d}(A) M_{d}(A) \Phi(A^{(d)}_{l})\Phi(A^{(d)}_{r}) 
	\end{equation}
	if $A$ is not atomic, and $\Phi(A) = 1$ if $A$ is atomic. 
\end{thm}
\noindent Remark: $\Phi(\om)$ is the overall marginal likelihood. It is a function of the hyperparameters $\bphi$, and can be used for specifying the hyperparameters $\bphi$ in an empirical Bayes strategy using maximum marginal likelihood estimation (MMLE).
\vspace{0.2em}

\vspace{0.5em}

\begin{thm}
	\label{thm:latent_state}
	Suppose $\T \sim {\rm RRDP}(\blam)$ and given $\T$ under a Haar DWT, one models the wavelet coefficients conditionally independently given a set of latent state variables $\S=\{S_{j,k}:j=0,1,2\ldots,J, k=0,1,\ldots,2^j-1\}$ 
	\[ (w_{j,k},z_{j,k})\,|\,S_{j,k}=s \ind p_{j,k}^{(s)}(w,z\,|\,\bphi) \quad \text{for all $(j,k)$} \]
	where $S_{j,k}\in \{1,2,\ldots,K\}$ is a latent state variable associated with $(j,k)$. Also, suppose the collection of all latent variables is modeled as a top-down Markov tree (MT) with transition kernel $\brho$, $\S \sim {\rm MT}(\brho)$, i.e.,
	\[
	{\rm P}(S_{j,k}=s'\,|\,S_{j-1,\lfloor k/2 \rfloor} = s) =\rho_j(s,s')
	\]
	where $\rho_j(\cdot,\cdot)$ is the transition kernel of the Markov model which is allowed to be different over~$j$. Then the joint marginal posterior of $(\T,\S)$ can be specified fully as the following sequential generative process. Suppose $\T^{0},\T^{1},\ldots,\T^{j}$ and the latent variables up to level $j-1$ have been generated. (To begin, we have $j=0$ and $\T^{0}=\{\om\}$.) Then the state variables in level $j$, are generated from the following posterior transition probabilities
	\begin{align} 
	& {\rm P}(S_{j,k}=s'\,|\,S_{j-1,\lfloor k/2 \rfloor} = s,\T^{(j)},\by) \\
	 = & \; \rho_{j}(s,s') \sum_{d \in \D(A)} \lambda_{d}(A) M_{d}^{(s')}(A) \Phi_{s'}(A^{(d)}_{l})\Phi_{s'}(A^{(d)}_{r})/\Phi_{s}(A), 
	\end{align}
	where $A$ is the node $A_{j,k}$ in $\T^{j}$. Given $S_{j,k}=s'$, suppose $j<J$, then $\T^{j+1}$ is generated by drawing $D_{j,k}$ from a multinomial with probabilities $\tilde{\blam}(A)$ such that
	\begin{align} 
	& {\rm P}(D_{j,k}=d\,|\,S_{j,k}=s',\T^{(j)},\by) \\
	= & \frac{\lambda_{d}(A) M_{d}^{(s')}(A) \Phi_{s'}(A^{(d)}_{l})\Phi_{s'}(A^{(d)}_{r})}{\sum_{d' \in \D(A)}\lambda_{d'}(A) M_{d'}^{(s')}(A) \Phi_{s'}(A^{(d')}_{l})\Phi_{s'}(A^{(d')}_{r})}, 
	\end{align}
	where $M_{d}^{(s)}(A)$ is the marginal likelihood contribution from the wavelet coefficient on node $A$ if it is a node in $\T$, is divided in dimension $d$ in $\T$, and its latent state is $s$. That is,
	$M^{(s)}_{d}(A) = \int p^{(s)}_{j,k}(w_{d}(A),z\,|\,\bphi) \, dz$
	and $\bm{\Phi}=(\Phi_1,\Phi_2,\ldots,\Phi_K):\A \rightarrow [0,\infty)^{K}$ is a vector-valued mapping defined recursively as 
	$
		\Phi_s(A) = \sum_{s'}\rho_j(s,s') \sum_{d\in \D(A)} \lambda_{d}(A) M^{(s')}_{d}(A) \Phi_{s'}(A^{(d)}_{l})\Phi_{s'}(A^{(d)}_{r})
	$
	if $A$ is not atomic, and $\Phi_s(A) = 1$ if $A$ is atomic,
	for all $s\in \{1,2,\ldots,K\}$, where $j$ is the level of $A$.
\end{thm}

Once the marginal posterior of $\T$ is computed through Theorem~\ref{thm:independent_shrinkage} or Theorem~\ref{thm:latent_state}, the full joint posterior is available as the conditional posterior of the rest of our model given $\T$ is available for common Bayesian wavelet regressions. (More details are given in Section~\ref{section:background}.)  Then standard Bayesian inference can proceed. 

In particular, one can draw samples for $(\T,\S)$ from their marginal posterior given  in Theorem~\ref{thm:latent_state}. Then given $(\T,\S)$, one can further sample $\bz$ from the conditional posterior corresponding to the chosen wavelet regression model, and Bayesian inference can proceed in the usual manner. For example, one can obtain posterior samples of the underlying function $\vect{f}$ by first drawing samples
\[(\T^{(1)},\S^{(1)},\bz^{(1)}),(\T^{(2)},\S^{(2)},\bz^{(2)}),\ldots,(\T^{(B)},\S^{(B)},\bz^{(B)}).\]
Then for the $b$th draw, we can compute the corresponding function $\vect{f}^{(b)}$ using the inverse DWT
\[
\vect{f}^{(b)} = \pi_{\T^{(b)}}^{-1} \left(W^{-1} \bz^{(b)}\right),
\]
where $\pi_{\T}^{-1}$ denotes the inverse permutation corresponding to an RDP $\T$. Based on the posterior samples of $\vect{f}$, we can construct pointwise credible bands and estimate the posterior mean ${\rm E}(\vect{f}|\by)$. 
We can apply Rao-Blackwellization and obtain the following estimate for the posterior mean
\[
{\rm E}(\bm{f}\,|\,\by) \approx \frac{1}{B} \sum_{b=1}^{B} \pi_{\T^{(b)}}^{-1} \left(W^{-1} \text{E}(\bz^{(b)} | \T_i^{(b)}, \by) \right). 
\]
For several popular Bayesian wavelet regression models, 
the posterior mean can actually be computed analytically through message passing (MP) without posterior sampling %\color{blue} 
when the Haar basis is adopted.
We next turn to briefly reviewing these wavelet regression models in Section~\ref{section:background}, and defer the MP algorithm (Theorem~\ref{thm:recursive.map}) to Supplementary Materials. \color{black}

\subsubsection{Examples of compatible Bayesian wavelet regression models\label{section:background}} 
So far we have kept the description of the Bayesian wavelet regression model general, using generic notations such as $p(w_{j,k},z_{j,k}\,|\,\bphi)$ and $p(w_{j,k},z_{j,k}\,|\,S_{j,k},\bphi)$ without spelling out the details. Next we describe some of the most popular Bayesian wavelet regression models. They indeed take these general forms and therefore our framework is applicable to them.

A popular class of Bayesian wavelet regression models for achieving adaptive shrinkage of $\bz$ utilize the so-called spike-and-slab prior, which introduces a latent binary random variable $S_{j,k}$ for each $(j,k)$ such that
\begin{align}
\label{eq:spike_and_slab}
z_{j,k}\,|\,S_{j,k} \ind (1 - S_{j, k}) \delta_0(z_{j,k}) + S_{j, k} \gamma(z_{j, k} | \tau_j, \sigma) 
\end{align}
where $\delta_0(\cdot)$ is a point mass at 0, and $\gamma(\cdot | \tau_j, \sigma)$ is a fixed unimodal symmetric density that possibly depends on $\sigma$ and another scale parameter $\tau_j$.  A common choice of $\gamma(\cdot | \tau_j, \sigma)$ is the normal distribution with mean 0 and variance $\tau_j \sigma^2$, denoted by $\phi(\cdot | 0, \sqrt{\tau_j} \sigma)$, while heavy-tailed priors including the Laplace and quasi-Cauchy distributions~\citep{Johnstone+Silverman:05} also enjoy desirable theoretical properties. Specifically, the function $\gamma(x \mid \tau_j, \sigma)$ is 
$$\gamma(x \mid \tau_j, \sigma) = a \exp(-a |x / \sigma|)/(2\sigma)$$
for Laplace priors where $a = \sqrt{2/\tau_j}$, and 
$$\gamma(x \mid \tau_j, \sigma) = (2 \pi)^{-1/2}\{1 - |x/\sigma| \cdot \tilde{\Phi}(|x/\sigma|)/\phi(x/\sigma)\}/\sigma$$
for quasi-Cauchy priors with $\tilde{\Phi}(x) = \int_x^{\infty} \phi(t \mid 0, 1) dt$.

Many authors \citep{Chipman1997,Clyde+George:00,Brown2001,Morris2006} adopt independent priors on the latent shrinkage state variable $S_{j,k}$
\[ S_{j,k}  \ind {\rm Bern}(\rho_{j,k}).\]
One way to specify $\brho = \{\rho_{j,k}, 0 \leq k < 2^j, 0 \leq j \leq J - 1\}$ that properly controls for multiplicity is $\rho_{j,k} \propto 2^{-j}$. The specification of $\btau = \{\tau_j, 0 \leq j \leq J - 1\}$ of course depends on the choice of $\gamma(\cdot | \tau_j, \sigma)$. For instance, if one uses $\tau_j=2^{-\alpha j}\tau_0$ for the normal and Laplace prior, this leads to the reduced parameter $\btau = (\alpha, \tau_0)$. One can use $\tau_j \equiv 1$ for the quasi-Cauchy prior. Other authors \citep{crouse:1998,ma&soriano:2018fanova} show that introducing Markov dependency into the latent shrinkage states can substantially improve inference by allowing effective borrowing of information across the location and scale. 

Carrying out inference under WARP requires the conditional posterior of $z_{j,k}$ given $(\T,\S)$. For the above popular models, this posterior is given by
\[
z_{j,k}\,|\,S_{j, k}, \by \ind (1 - S_{j, k}) \delta_0(z_{j, k}) + S_{j, k} f_1(z_{j, k} | w_{j,k}, \tau_j, \sigma),
\]
where $f_1(z_{j, k} \mid w_{j, k}, \tau_j, \sigma) \propto \phi(w_{j, k} \mid z_{j, k}, \sigma) \cdot \gamma(z_{j, k} \mid \tau_j, \sigma)$. The function $f_1(z_{j, k} \mid w_{j, k}, \tau_j, \sigma)$ is analytically available if $\gamma(\cdot \mid \tau_j, \sigma)$ is the density of normal, Laplace, or quasi-Cauchy distributions. 
For the normal prior where $\gamma(\cdot \mid \tau_j, \sigma) = \phi(\cdot \mid 0, \sqrt{\tau_j} \sigma)$, $f_1(\cdot \mid w_{j, k}, \tau_j, \sigma)$ is the density of ${\rm N}(w_{j, k}/(1 + \tau_j^{-1}), \sigma^2/(1 + \tau_j^{-1}))$. For Laplace and quasi-Cauchy priors, analytical forms of $f_1(\cdot \mid \tau_j, \sigma)$ are available in~\cite[Sec.~2.3]{Johnstone+Silverman:05}.  As it is often the mean corresponding to $f_1$ that is needed for posterior estimation, we here give the closed forms of the means by integrating out $z_{j, k}$ with respect to its posterior distribution. Let the corresponding mean function be $\mu_1(w_{j, k}, \tau_j, \sigma)$, which is given by 
\[
w_{j, k}/(1 + \tau_j^{-1})
\]
for normal priors, 
\[
w_{j, k} -  \sigma \frac{a \{ e^{-a w_{j, k}/\sigma} \Phi(w_{j, k}/\sigma - a) - e^{aw_{j, k}/\sigma} \tilde{\Phi}(w_{j, k}/\sigma + a) \}}{e^{-a w_{j, k}/\sigma} \Phi(w_{j, k}/\sigma - a) + e^{aw_{j, k}/\sigma} \tilde{\Phi}(w_{j, k}/\sigma + a)}
\]
for Laplace priors, and 
\[
{w_{j, k}} \left\{1 - \exp\left(-\frac{w_{j, k}^2}{2\sigma^2}\right)\right\}^{-1} - 2 \left( \frac{w_{j, k}}{\sigma}\right)^{-1}
\]
for quasi-Cauchy priors.

For these wavelet regression models that adopt the spike-and-slab setup, by Theorem~\ref{thm:latent_state} we can derive a fully conjugate posterior that takes the same form as the prior. In particular, for each $A \in \A$, under the normal prior for $\gamma(\cdot \mid \tau_j, \sigma)$, applying Theorem~\ref{thm:latent_state} shows that
\begin{itemize}
	\item The marginal likelihood contribution from the data within node $A$ if $A$ is divided in dimension $d$ is: 
$$M_{d}^{(s)}(A)=\frac{1}{\sqrt{2\pi(1+s\tau_j)\sigma^2}} \exp\left\{-\frac{w_{d}(A)^2}{2\sigma^2(1+s\tau_j)}\right\}$$ for $s=0,1.$ 
	
	\item The posterior spike probability on $A$ if $A$ is divided in dimension $d$ is:  $$\tilde{\rho}_{d}(A)=\rho(A)M^{(1)}_{d}(A)/M_{d}(A),$$
	where $M_{d}(A)=\rho(A) M_{d}^{(1)}(A) + (1-\rho(A)) M_{d}^{(0)}(A).$
\end{itemize}  

In most practical problems, the variation in the function value within each partition block will eventually become negligible with respect to the noise level, and so further division within such homogeneous blocks will not improve statistical efficiency and could lead to overfitting. For example, in~\ref{fig:RDP} the partition in the upper left block (Level 3) along with its descendants is not necessary. Thus it is also desirable to incorporate adaptivity in the depth of the wavelet tree along each subbranch and allow it to be terminated earlier than reaching level $J$ depending on how smooth the function is across the index space. This consideration is closely related to the idea of adaptive block shrinkage \citep{Cai1999} in the frequentist wavelet regression analysis. Once there is little evidence for any interesting structure within a subset of the index space, then the function value within that subset can be shrunk to a constant. That is, the wavelet tree is ``pruned'' there. Remarkably, wavelet models with such pruning are also compatible with our WARP framework and can be readily achieved by introducing a pruning indicator to accompany $S_{j, k}$. We refer interested readers to Supplementary Materials for additional technical details on how to incorporate pruning.

For the Haar basis, the posterior mean ${\rm E}(\vect{f}|\by)$ for the above wavelet models can be evaluated analytically through recursive message passing without any Monte Carlo sampling for Bayesian wavelet regression models that adopt the spike-and-slab setup along with optional pruning of the wavelet tree, which contains the models without optional pruning as special cases with zero pruning probabilities. For completeness, we describe this strategy in the Supplementary Materials and will use it to compute ${\rm E}(\vect{f}|\by)$ in our numerical examples.
\color{black}

\section{Experiments\label{sec:experiments}}

In this section, we conduct extensive experiments to evaluate the performance of our proposed framework in the image reconstruction task in terms of both estimation accuracy and computational scalability. Applications to other image processing tasks are discussed in \ref{sec:energy}. We compare WARP to a number of state-of-the-art wavelet, non-wavelet, and deep neural network-based methods available in the literature for denoising 2D images. We provide results on denoising 3D images in the Supplementary Materials. Throughout these experiments we apply WARP with independent spike-and-slab Bayesian wavelet regression models under the Haar basis along with optional pruning. 
\color{black} 

Our prior specification is as follows: $\rho(A) = \min(1, 2^{-\beta j} C)$ for $A$ in the $j$th resolution (for $j<J$), $\tau_j = 2^{-\alpha  j} \tau_0$, and $\eta(A) = \eta_0$ for all $A$; we set $\sigma^2$ to an estimate based on the finest scale wavelet coefficients~\citep{donoho1995adapting}; all other parameters in $\bphi = (\alpha, \beta, \sigma^2, \tau, C, \eta_0)$ are estimated by maximizing the marginal likelihood (available in a closed form as $\Phi(\om)$ from our recursive message passing algorithm) at a set of grid points. Supplementary Materials contain a sensitivity analysis showing that WARP is generally robust to the values of its hyperparameters. Therefore we recommend a grid search on a small set rather than a full optimization as the default tuning method. Gaussian noise with standard deviation $\sigma$ is added to the true images and we apply all methods to the noisy observations for image reconstruction. For WARP, we use the posterior mean as the reconstructed image, which is analytically attainable through Theorem~\ref{thm:recursive.map}. 

\subsection{Image reconstruction using ImageNet data\label{sec:2D} }
We use 100 test images randomly chosen from the ImageNet dataset~\citep{deng2009imagenet} to evaluate selected methods in reconstructing images of various structures. ImageNet is originally used for large-scale visual recognition in the community of computer vision, and we here use its Fall 2011 release (consisting of 14,197,122 urls). We compare our method with eight existing wavelet and non-wavelet approaches with available software: 1-dimensional Haar denoising operated on vectorized observation~\citep{Johnstone+Silverman:05} or 1D-Haar, translation-invariant 2D Haar estimation~\citep{Wil+Now:04} or TI-2D-Haar, shape-adaptive Haar wavelets~\citep{fryzlewicz2016shah} or SHAH, adaptive weights smoothing~\citep{polzehl2000adaptive} or AWS, Bayesian smoothing method using the Chinese restaurant process~\citep{Li+Ghosal:14} or CRP, coarse-to-fine wedgelet~\citep{Cas+:04} or Wedgelet, nonparametric Bayesian dictionary learning proposed by~\cite{Zhou+:12} or  BPFA, and the conventional running median method or RM. We apply the cycle spinning technique to remove visual artifacts in image reconstruction~\citep{Coi+Don:95, Li+Ghosal:15} for the methods of WARP, 1D-Haar, SHAH, AWS, CRP, Wedgelet and RM, by averaging 121 local shifts (a step size up to 5 pixels in each direction). TI-2D-Haar is translation invariant and BPFA includes cycle spinning based on patches, and thus no additional cycle spinning is needed for these two methods. For each method, we calculate the mean squared error (MSE) and mean absolute error (MAE) to measure its accuracy, and time each method based on one replication ran on MacBook Pro with 2.7 GHz Intel core i7 CPU and 16GB RAM. We implement the methods using publicly available code, either in R (1D-Haar, SHAH, and AWS) or Matlab (TI-2D-Haar, CRP, Wedgelet, BPFA, and RM). WARP is available in both R and Matlab, and we use the R version to time it. 

\ref{figure:imagenet} presents the average MSEs and MAEs of all methods where $\sigma$ varies from 0.1 to 0.7. We can first see that the proposed hierarchical adaptive partition improved the basic wavelet regression significantly (compare 1D-Haar vs. WARP) for all scenarios. 
In fact, WARP is uniformly the best method under both metrics for all scenarios, with the performance lead over other methods widening as the noise level increases. The sensitivity analysis in the Supplementary Materials indicates that the method of WARP is robust to hyperparameters and choices of $\gamma$. 

\begin{figure}
	\centering
	\begin{tabular}{cc}
		\includegraphics[page = 1, width = 0.45\linewidth]{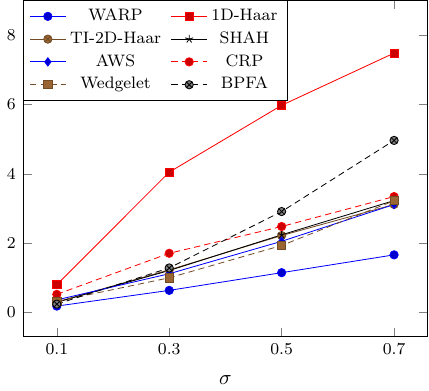} &
		\includegraphics[page = 2, width = 0.45\linewidth]{figure/plotCurve.pdf} \\
		(a) MSE ($\times 10^{-3}$)  & 
		(b) MAE ($\times 10^{-2}$) 
	\end{tabular} 
	\caption{\label{figure:imagenet}Comparison of various methods based on 100 randomly selected $512 \times 512$ images from ImageNet. The method of running median is off the chart (not plotted here). The maximum standard errors at each $\sigma$ among all methods are $(0.001,0.042,0.071,0.058)\times10^{-3}$ for MSE, $(0.002,0.062, 0.065, 0.058) \times 10^{-2}$ for MAE, respectively. The running time of each method in seconds is 7.2 (WARP), 76.9 (SHAH), 7.9 (AWS), 10.7 (CRP), 8.7 (Wedgelet), $2.1 \times 10^3$ (BPFA), and less than 1 (1D-Haar, TI-2D-Haar, RM), based on one test image without cycle spinning at $\sigma$ = 0.3 including both tuning and estimation steps.} 
\end{figure}

WARP is computationally efficient, benefiting from the conjugacy of random recursive partition and closed form expression in Theorem~\ref{thm:recursive.map}. WARP is the fastest adaptive approach among SHAH, AWS, CRP, Wedgelet, and BPFA. (The computing times are given in the caption of \ref{figure:imagenet}.) Section \ref{sec:scalability} further compares the scalability of selected methods using images of various sizes.

\newcommand{\plotF}[1]{
	\includegraphics[width = 0.2 \linewidth, trim = 20 0 0 0, clip]{figure/obs#1/true1type21d.png} & 
	\includegraphics[width = 0.3 \linewidth, trim = 20 0 0 0, clip]{figure/obs#1/saving1type21d.png} & 
	\includegraphics[width = 0.3\linewidth, trim = 20 0 0 0, clip]{figure/obs#1/saving1type22d.png}\\}

\begin{figure*}
	\setlength{\tabcolsep}{3pt}
	\centering
	\begin{tabular}{ccc}
		\plotF{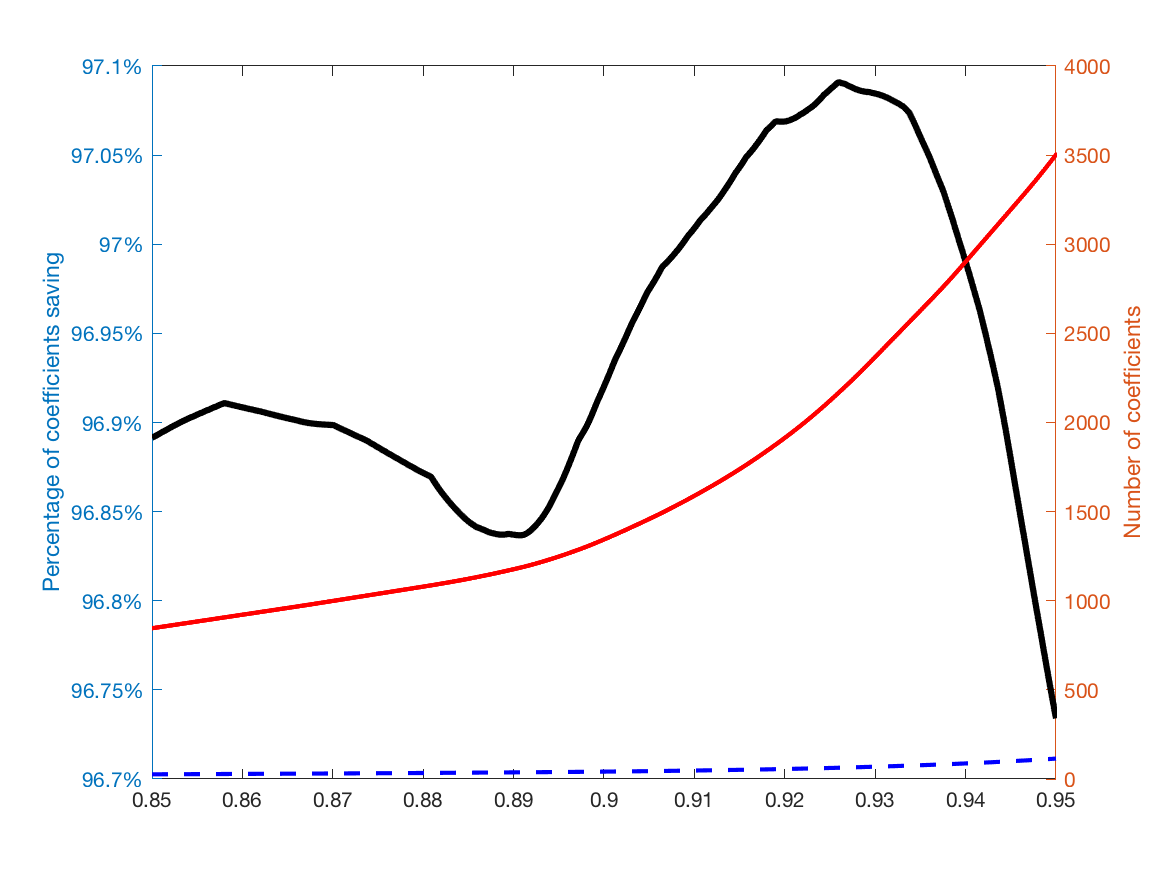} \plotF{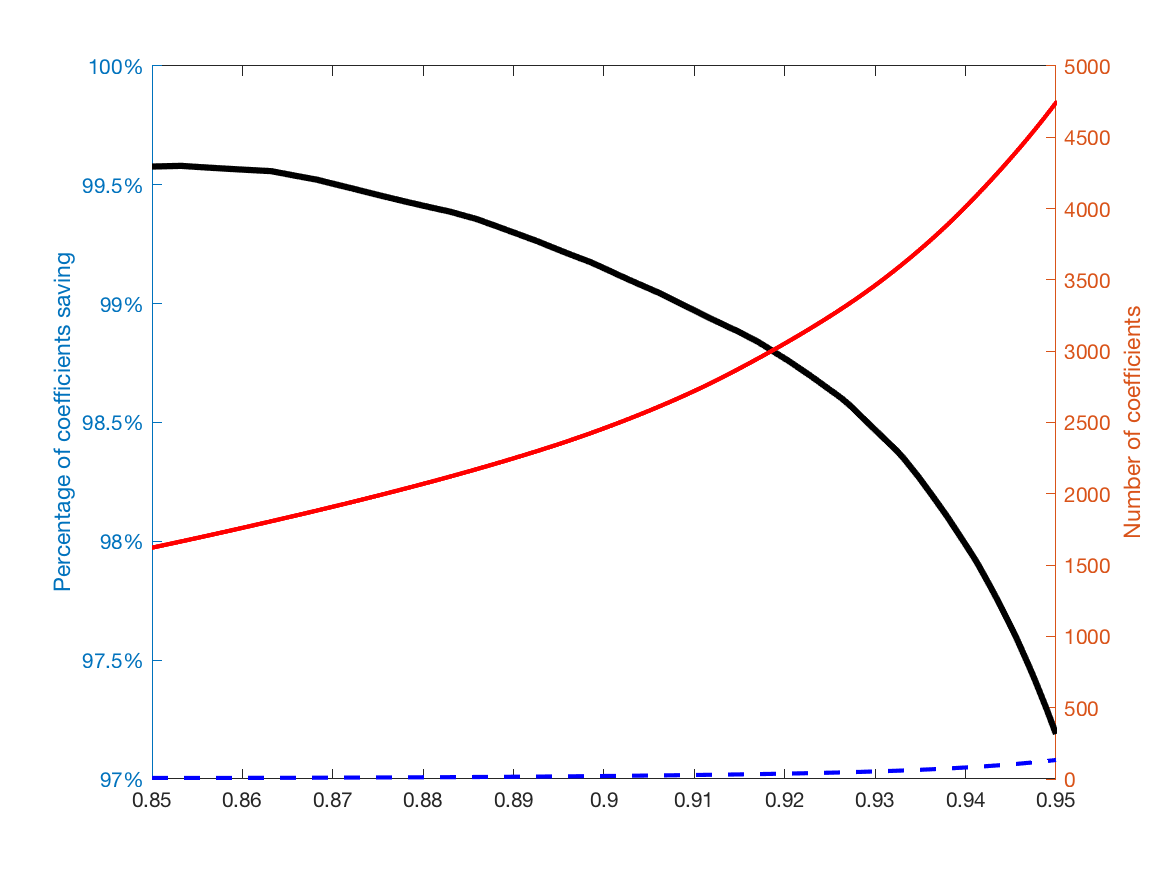} \plotF{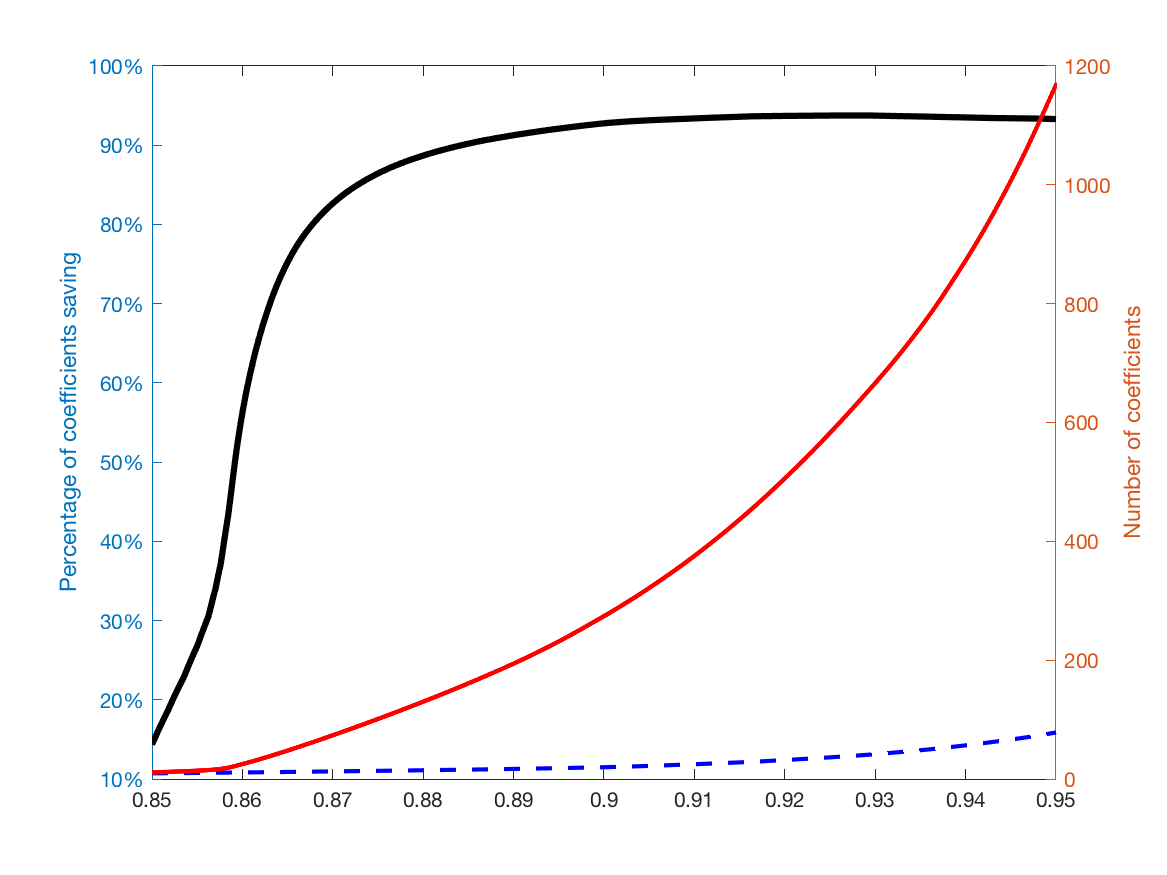} 
		(a) true & (b) WARP vs 1D-DWT  & (c) WARP vs 2D-DWT  \\
	\end{tabular} 
	\caption{Comparison of energy concentration for three methods---WARP,1D Haar, and 2D Haar---on ImageNet images. Column (a) plots the true image, Column (b) compares WARP versus 1D DWT, and Column (c) compares WARP versus 2D DWT. In Columns (b) and (c), the red and blue lines correspond to the right $y$ axis, plotting the number of coefficients to attain a specific energy level ($x$ axis) by deterministic DWT and WARP, respectively. The black curve corresponds to the left $y$ axis and is 100\% less the ratio of the blue and red curves, indicating the percentage reduction in the number of wavelet coefficients to achieve the same sum of squares by WARP.}
	\label{fig:energy} 
\end{figure*}

\subsection{Scalability\label{sec:scalability}} 
Next we verify the linear complexity of the WARP framework using both 2D and 3D images. Usually there are various ways to tune each method, and we focus on the estimation step given tuning parameters for all methods to make a fair comparison. For WARP, one actually may choose the tuning parameters from a smaller image by downsampling without loss of much accuracy, in view of its insensitivity to hyperparameters (Section D in the Supplementary Materials). 

\ref{fig:scalability} (a) compares the scalability of selected methods in~\ref{figure:imagenet}; we exclude 1D-Haar and RM as their reconstructions are highly inaccurate, and BPFA as it scales poorly even at $512 \times 512$ images. We can see that the empirical running time approximately follows a linear function of the number of locations. In fact, WARP takes only about 2 minutes for a large image of $4096 \times 4096$ that contains 17 million pixels, and 5.3 seconds for an image of 1024 by 1024. \ref{fig:scalability} suggests that Wedgelet and SHAH take quadratic time or even more, while TI-2D-Haar, AWS, and CRP takes linear time, but their performances are substantially inferior to that of WARP as shown in \ref{figure:imagenet}. CRP seems to have a smaller slope than WARP, but it requires considerably longer tuning time than WARP according to the total running time with the tuning step included in the caption of~\ref{figure:imagenet}, at least based on its latest version of implementation to date. 

It is worth noting that while many state-of-the-art methods designed for 2D images such as Wedgelet, TI-2D-Haar, and BPFA require substantial modifications for a new dimensional setting, such as 3D images, the proposed WARP framework is directly applicable to $m$-dimensional data without modification, with the same linear scalability as suggested by \ref{fig:scalability} (b). 

\begin{figure}
	\begin{tabular}{cc}
		\includegraphics[page = 5, width = 0.48\linewidth]{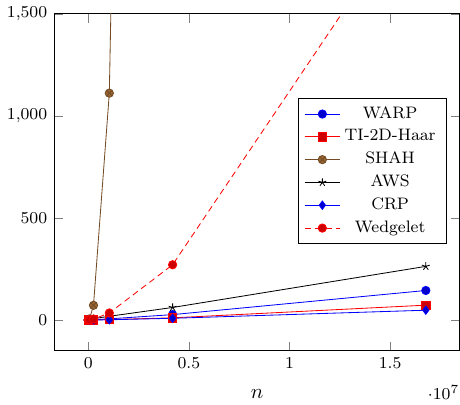}& 
		\includegraphics[page = 6, width = 0.46\linewidth]{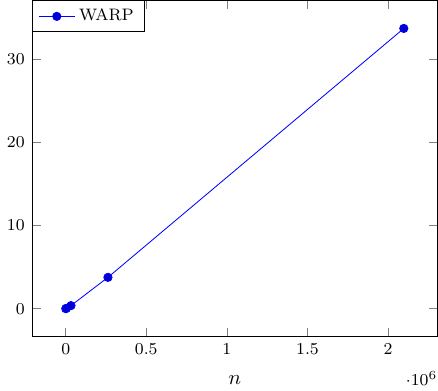} \\
		(a) 2D images & (b) 3D images 
	\end{tabular}
	\caption{Scalability of various methods using 2D and 3D images. Each line is the running time taken by the estimation step ($y$-axis) using the corresponding method versus the number of locations in the image ($x$-axis). }
	\label{fig:scalability} 
\end{figure}

%\color{blue} 
\subsection{Comparison with deep neural networks} \label{sec:DnCNN} 
In this section, we compare the proposed method WARP with convolutional neural networks. In particular, we apply WARP and the denoising convolutional neural networks (DnCNN) proposed in \citep{Zhang2017} to two popular benchmark datasets: the twelve widely used test images (\ref{fig:set12} in Supplementary Materials) and the BSD68 data~\citep{MartinFTM01} which contains 68 natural images from the Berkeley segmentation dataset. DnCNNs have been reported to achieve the state-of-the-art performance in various image processing tasks \citep{Zhang2017}. We adopt a pre-trained model available in Matlab for DnCNN. 

\ref{table:DnCNN} reports the MSEs of WARP and DnCNN on the 12 widely used test images and BSD68 (averaged) at three noise levels when $\sigma = 0.2, 0.4, 0.6$. We can see that for light noise when $\sigma = 0.2$, WARP leads to smaller MSEs in five out of twelve images (i.e., Image 2, 3, 6, 8, 12) and gives comparable performance in other images. WARP gives uniformly smaller MSEs when $\sigma = 0.4$ (intermediate noise) and constantly outperforms DnCNN by one order of magnitude when $\sigma = 0.6$ (large noise), which is consistent with our observations in the ImageNet experiment. Besides the excellent performance of WARP, it is worth mentioning that unlike DnCNN which requires substantially more extensive pre-training and tuning, WARP does not require pre-training at all, and its small amount of tuning can be completely automated without user intervention.  We do acknowledge that the performance of the pre-trained DnCNN might be improved with more extensive training.

\begin{table*}[ht]
	\centering
	\caption{MSE ($\times 10^{-3}$) of WARP and DnCNN on 12 widely used test images and BSD68.} \label{table:DnCNN}
	\begin{tabular}{ccccccccccccccr}
		\hline
		$\sigma$ &  & \multicolumn{12}{c}{12 widely used test images} & BSD68 \\
		 &  & 1 & 2 & 3 & 4 & 5 & 6 & 7 & 8 & 9 & 10 & 11 & 12 &  \\ 
		\hline
		\multirow{2}{*}{$0.2$}
		& WARP & 2.89 & 1.42 & 2.49 & 4.07 & 3.65 & 3.39 & 3.17 & 1.69 & 4.33 & 2.55 & 2.50 & 2.59 & 3.75 \\ 
		& DnCNN & 2.77 & 1.65 & 2.55 & 3.37 & 2.90 & 3.49 & 2.88 & 1.73 & 3.98 & 2.44 & 2.33 & 2.68 & 3.35 \\ \hline 
		\multirow{2}{*}{$0.4$}
		& WARP & 5.69 & 3.15 & 5.57 & 7.97 & 8.21 & 6.01 & 6.64 & 3.15 & 6.38 & 4.46 & 4.16 & 4.67 & 6.10 \\ 
		& DnCNN & 15.64 & 14.14 & 15.39 & 15.93 & 15.85 & 16.56 & 15.53 & 13.28 & 15.37 & 14.57 & 13.49 & 14.69 & 14.90 \\ \hline 
		\multirow{2}{*}{$0.6$}
		& WARP & 8.23 & 4.31 & 8.12 & 10.86 & 12.26 & 8.36 & 9.57 & 4.40 & 7.96 & 6.06 & 5.66 & 6.22 & 7.88 \\ 
		& DnCNN & 75.83 & 71.12 & 73.34 & 73.49 & 71.73 & 77.06 & 75.16 & 70.31 & 71.31 & 72.75 & 69.96 & 71.46 & 71.65 \\ 
		\hline
	\end{tabular}
\end{table*}
\color{black} 

\section{Enhanced energy concentration and beyond 2D image reconstruction \label{sec:energy} } 
The excellent performance of WARP in image reconstruction suggests that the model is capable of identifying efficient representation of the underlying structure in a variety of real images as it is designed to achieve. This also suggests that extracting the underlying representation can potentially benefit a variety of other downstream processing tasks. In this section we first use a concept of ``energy concentration" to examine how such efficiency is achieved and then discuss the potential applicability in other image processing tasks such as compression.

Energy (or information) concentration under wavelet transforms can be quantified by the number of wavelet coefficients needed to retain a given proportion of the sum of squares of the underlying function. An efficient wavelet representation will only need a small number of coefficients to capture most of the information contained in the function (as measured in terms of its sum of squares). Such a representation leads to high signal-to-noise ratios on a small number of coefficients that will facilitate all downstream processing tasks.

Next we compare energy concentration under WARP to that under classical 1D and 2D wavelet representations to quantify the improvement in energy concentration WARP achieves through adaptively identifying good permutations. To this end, we use the same ImageNet data as used  in~\ref{sec:experiments}. For each image, we draw a sample from the posterior distribution of partition trees produced by WARP, and compute the number of coefficients required to attain a range of energy levels (i.e., the total sum of squares) on a noisy observation at $\sigma = 0.1$ and compare them to those required under standard 1D and 2D wavelet transforms. 
\ref{fig:energy} presents the numbers of wavelet coefficients required over the proportion of the sum of squares for three representative images. 

Focusing on the proportion of the sum of squares from 0.85 to 0.95, we can see that the adaptive representation achieved by WARP requires substantially fewer numbers of wavelet coefficients (the red solid lines with scales on the right of each plot) to attain the same energy level than traditional 1D and 2D Haar DWT (the blue dashed lines with scales on the right of each plot).  In \ref{fig:energy}, we also plotted the percentage reduction in the number of coefficients (the black line with scales on the left of each plot) at each energy level. The largest coefficient saving of WARP is (80\%, 70\%, 70\%) compared to 2D DWT, and this saving becomes (97\%, 99\%, 90\%) when compared to 1D DWT. Enhanced energy concentration of WARP is observed in a wide range of test images in the database, and the extent of improvement in energy concentration varies according to the abundance of asymmetric structures present in the underlying image.

The improved energy concentration of WARP is expected to benefit a variety of downstream processing tasks beyond image denoising. For example, efficient image compression can be achieved using the posterior mode of the WARP model, which provides a sparse coding of the image. Coupling this idea with a pair of encoder and decoder, we introduce an algorithm for efficient image and video compression in a follow-up paper~\cite{Liu2020CARP}. Interested readers may refer to that paper for additional numerical experiments involving a variety of datasets, including 2D ImageNet, 3D medical image, real-life YouTube videos, and surveillance videos, in which
WARP-based compression substantially outperforms several state-of-the-art compression approaches.

\section{Application to retinal optical coherence tomography\label{sec:application}} 
We apply the proposed method to a dataset of optical coherence tomography (OCT) volumes. OCT provides a non-invasive imaging modality to visualize cross-sections of tissue layers at micrometer resolution, and thus is instrumental in various medical applications especially for the diagnosis and monitoring of patients with ocular diseases~\citep{huang1991optical,grewal2013diagnosis,virgili2015optical, cuenca2018cellular}. The accurate interpretation of OCT images may require the involvement of both retina specialists and comprehensive ophthalmologists, and this task is compounded by heavily noised observations at a low signal-to-noise ratio due to sample-based speckle and detector noise~\citep{keane2012evaluation,shi2015automated, Fang2017}. Therefore, reconstruction of OCT images is necessary to improve both manual and automated OCT image analysis, and is increasingly important when OCT images are used to extract objective and quantitative assessment in ophthalmology which is touted as one advantage of OCT in clinical practice~\citep{virgili2015optical}. 

We use the OCT data available at \url{http://people.duke.edu/~sf59/Fang_TMI_2013.htm}, acquired by a Bioptigen SDOCT system (Durham, NC, USA) at an axial resolution of $\sim$ 4.5 $\mu$m. We apply the methods of TI-2D-Haar, SHAH, AWS, CRP, Wedgelet, BPFA, and WARP, to two noisy slices (plotted as ``Obs." in~\ref{figure:OCT}). We also have access to a registered and averaged image by averaging 40 repeatedly sampled scans~\citep{Fang2017}, which is referred to as the ``noiseless'' reference image and is used to compare the quality of reconstructed images. From the results in~\ref{figure:OCT}, we clearly see that WARP gives the best global qualitative metric using MSE and MAE among all methods in comparison. 

Visual comparison provides a detailed assessment of reconstructed images on local features that might be clinically relevant. For the first observation in~\ref{figure:OCT}, we can see WARP distinguishes all layers well (the boxed region in the noiseless image), especially compared to TI-2D-Haar and AWS whose reconstructions are blurred across layers. For the second observation, we observe a separation of the posterior cortical vitreous from the internal limiting membrane in the noiseless image, which shows the potential to progress to vitreomacular traction (VMT)~\citep{duker2013international}. This separation becomes less clear if using TI-2D-Haar (especially the left proportion), although TI-2D-Haar gives MSE and MAE that are closer to WARP than the other methods. For both observations, there is still substantial noise left in the denoised images by SHAH, and AWS gives a reconstruction exhibiting undesirable patches. This study confirms that WARP is capable to denoise images while keeping important features present in the image, due to its ability to adapt to the geometry of the underlying structures. 

\newcommand{\OCTalpha}{2}
\newcommand{\OCTbeta}{15}
\begin{figure*}[!h] 
	\caption{Two retinal OCT datasets (titled ``Obs.") and reconstructed images using TI-2D-Haar, SHAH, AWS, CRP, Wedgelet, BPFA, and WARP. The two metrics following each method are the MSE ($\times 10^{-4})$ and MAE ($\times 10^{-2})$ respectively. The ``noiseless" reference is an registered and averaged image. }
	\label{figure:OCT}
	\begin{tabular}{c@{\hskip 0.1in}c@{\hskip 0.1in}c}
		\includegraphics[width=0.31\linewidth]{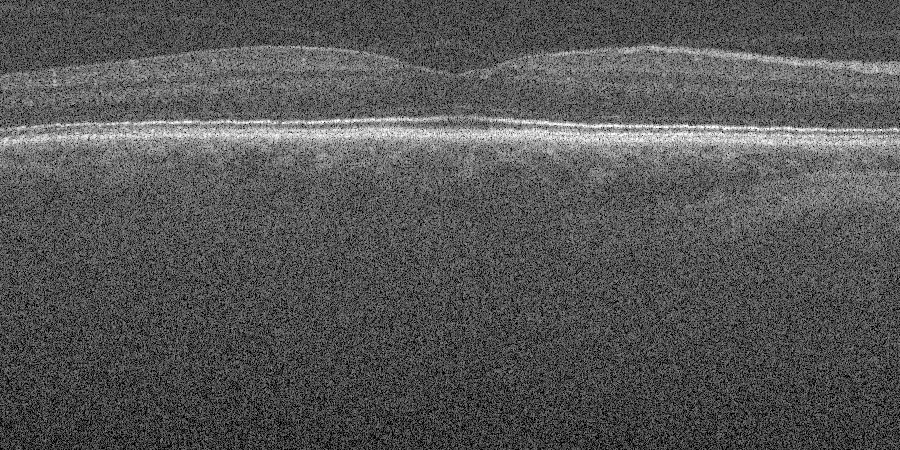} &
		\includegraphics[width=0.31\linewidth]{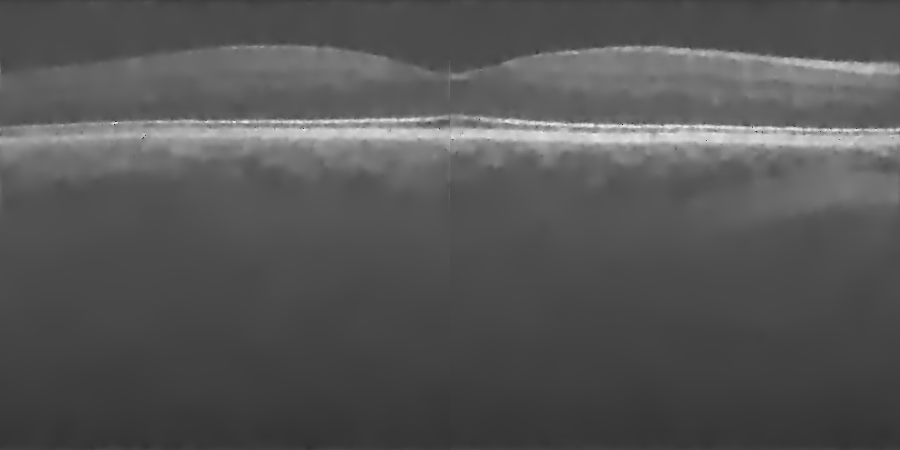} &
		\includegraphics[width=0.31\linewidth]{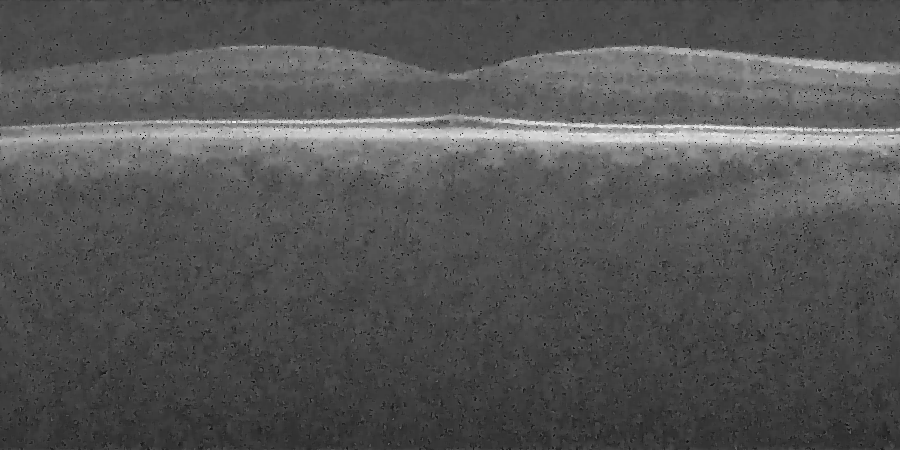} \\
		Obs. (152.4, 9.9) & TI-2D-Haar (7.3, 2.0) & SHAH (15.7, 2.7) \\
		
		\includegraphics[width=0.31\linewidth]{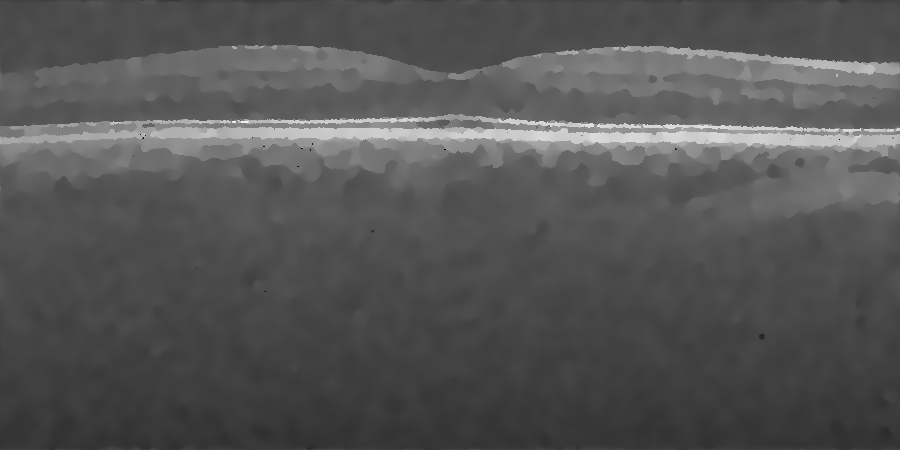} &
		\includegraphics[width=0.31\linewidth]{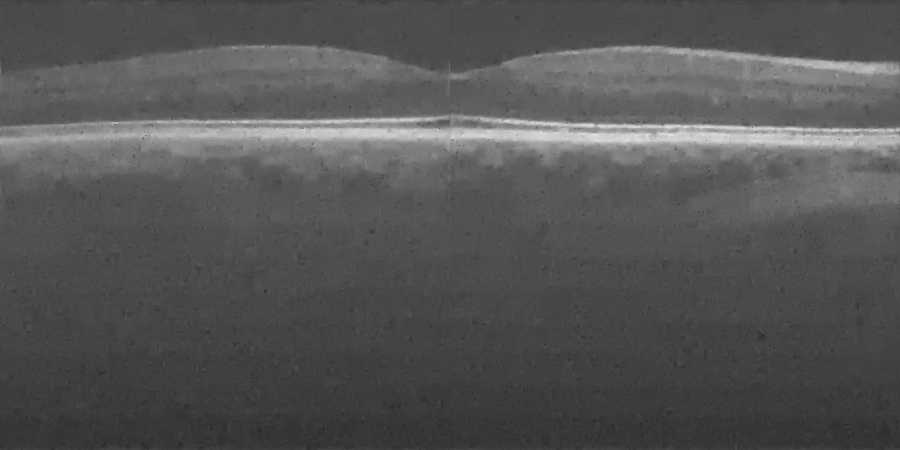} &
		\includegraphics[width=0.31\linewidth]{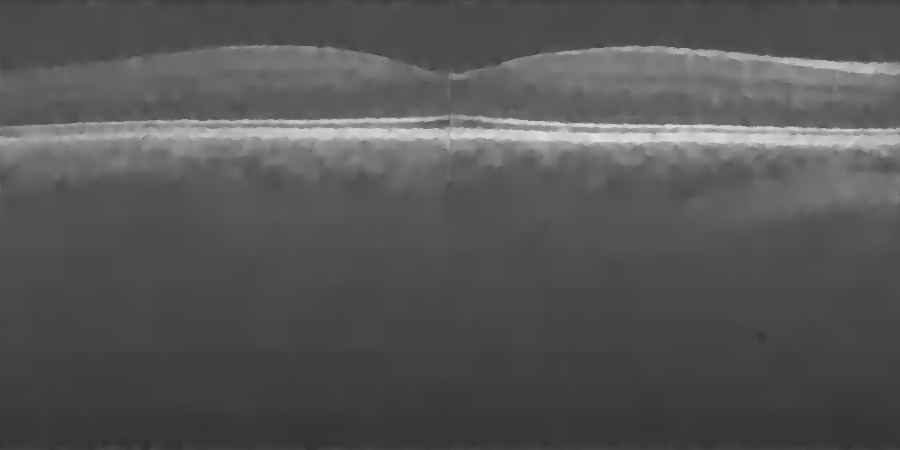} \\
		
		AWS (9.4, 2.2) & CRP (7.4, 2.1) & Wedgelet (7.7, 2.1) \\
		
		\includegraphics[width=0.31\linewidth]{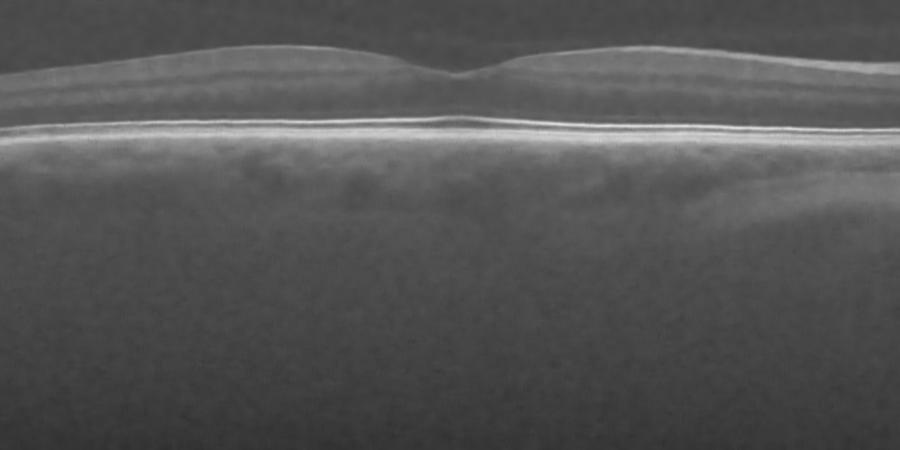} &
		\includegraphics[width=0.31\linewidth]{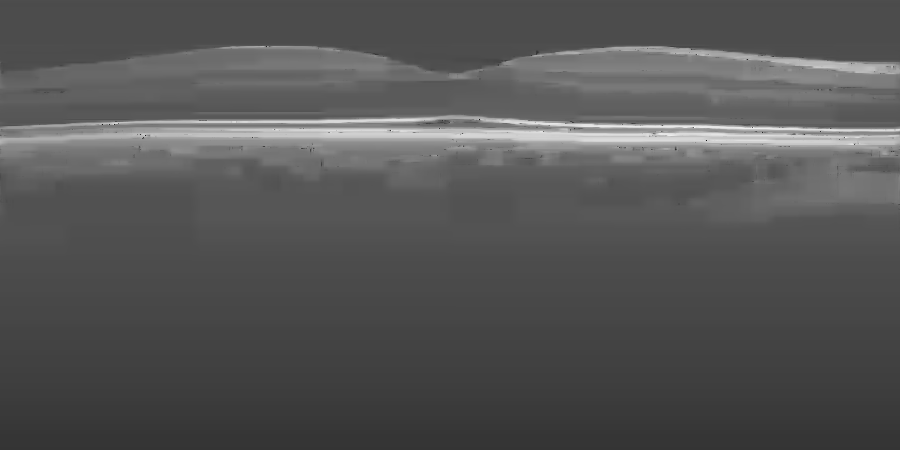} &
		\includegraphics[width=0.31\linewidth]{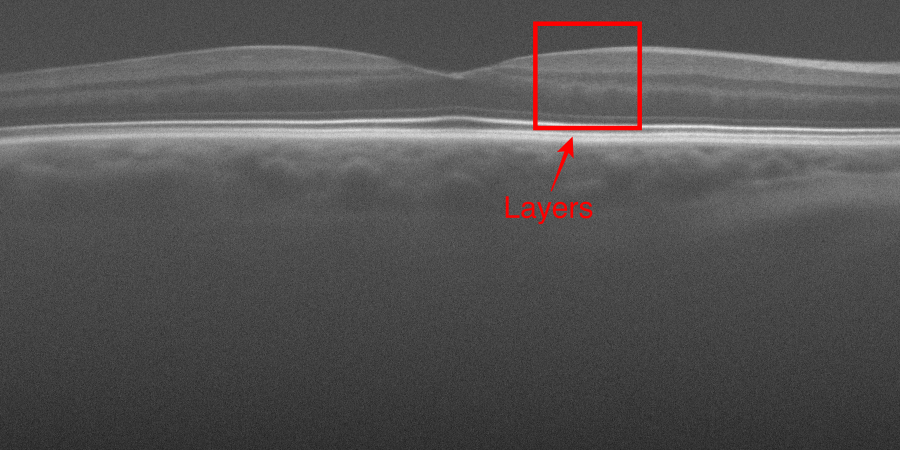} \\
		\vspace{0.2in}  
		BPFA (7.8, 2.2) & WARP (6.5, 1.9) & ``noiseless" \\

		\includegraphics[width=0.31\linewidth]{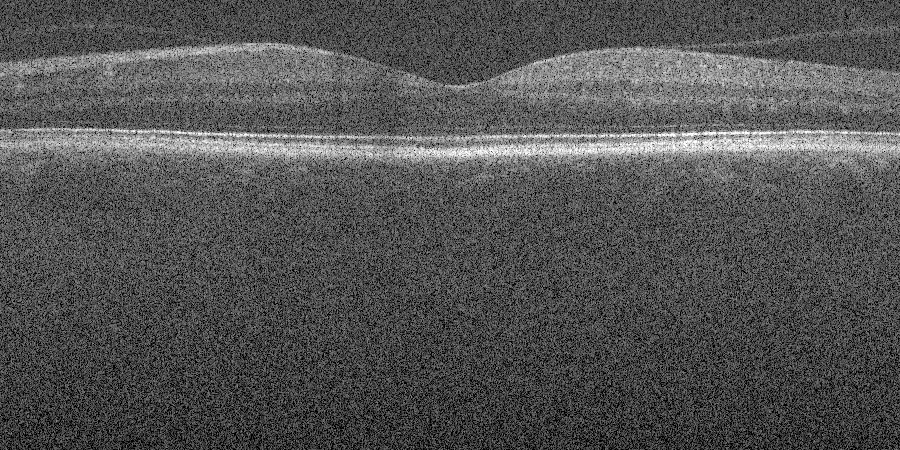} &
		\includegraphics[width=0.31\linewidth]{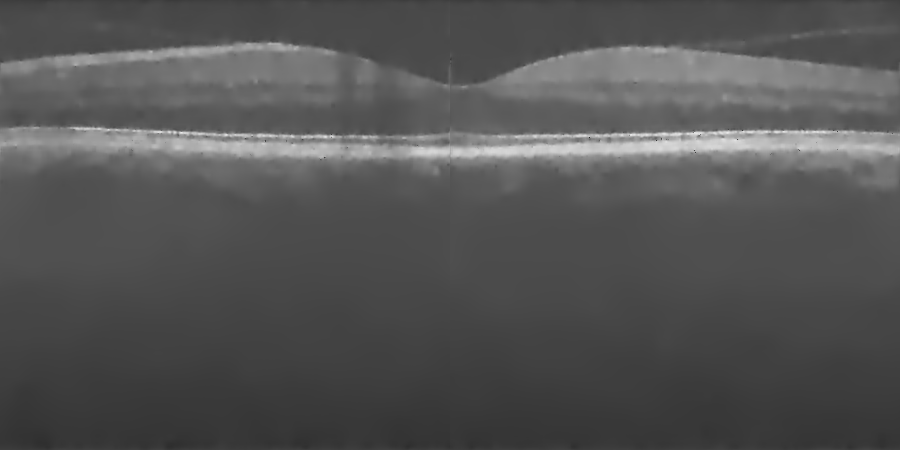} &
		\includegraphics[width=0.31\linewidth]{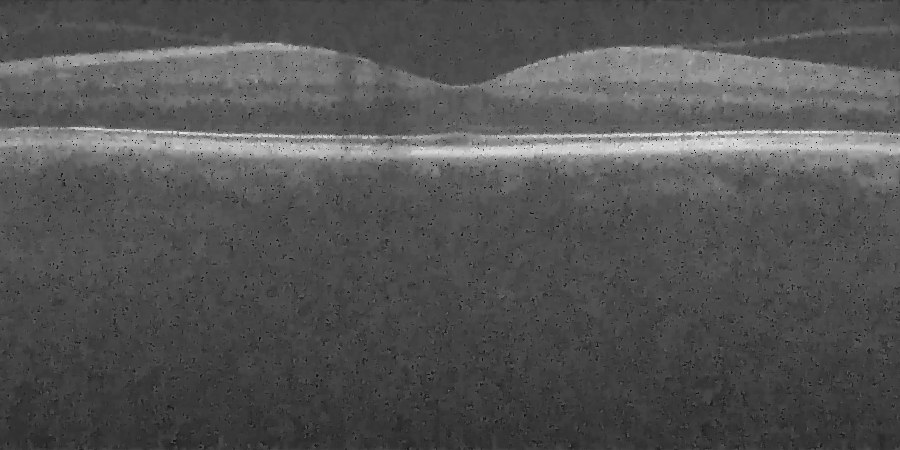} \\
		
		Obs. (159.4, 10.1) & TI-2D-Haar (10.9, 2.4) & SHAH (18.4, 3.0) \\
		
		\includegraphics[width=0.31\linewidth]{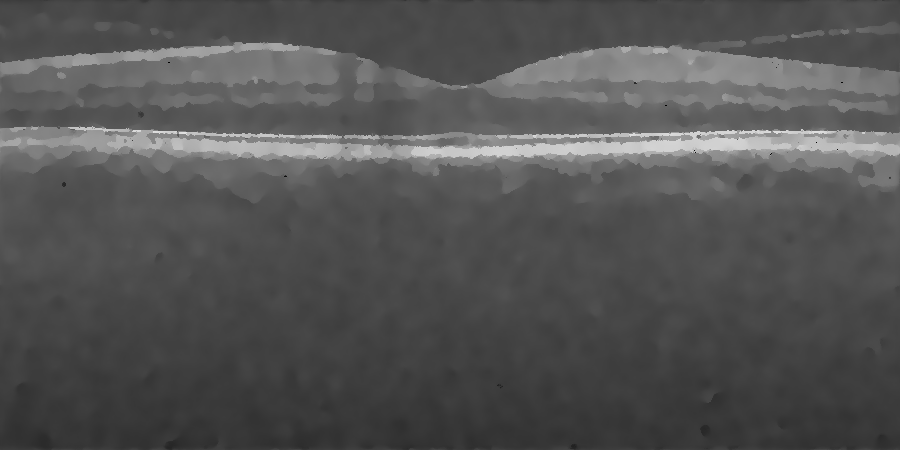} &
		\includegraphics[width=0.31\linewidth]{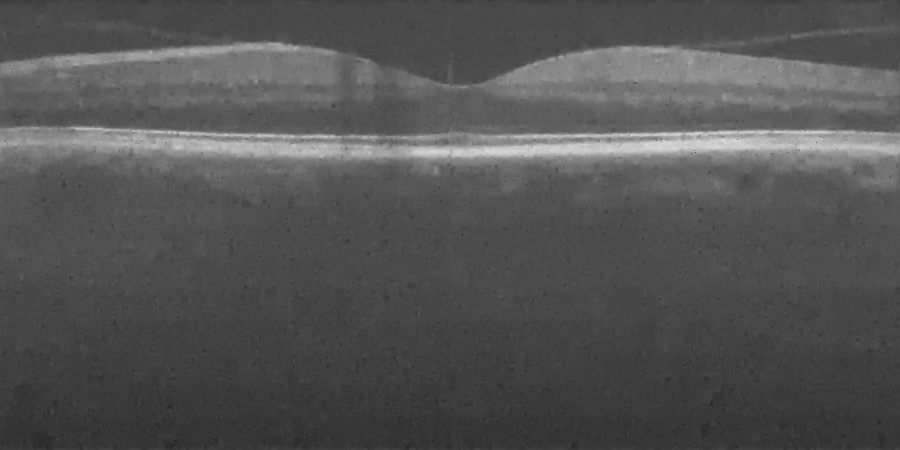} &
		\includegraphics[width=0.31\linewidth]{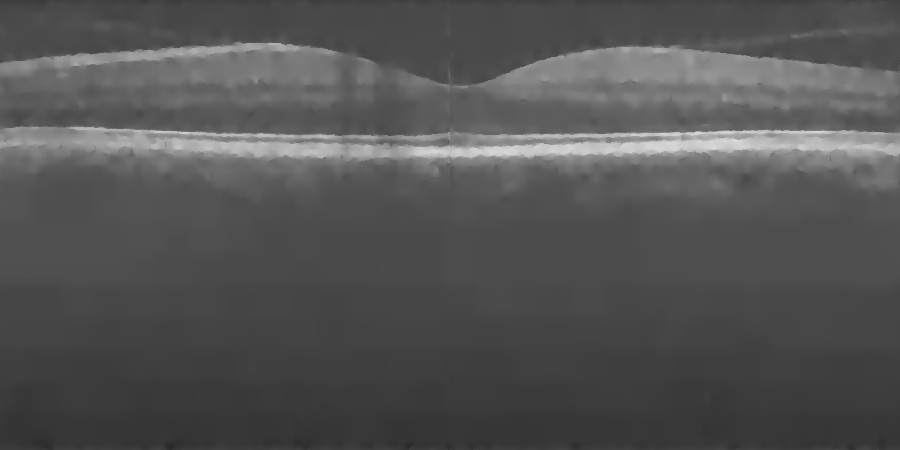} \\
		
		AWS (12.9, 2.5) & CRP (11.5, 2.5) & Wedgelet (11.1, 2.4) \\
		
		\includegraphics[width=0.31\linewidth]{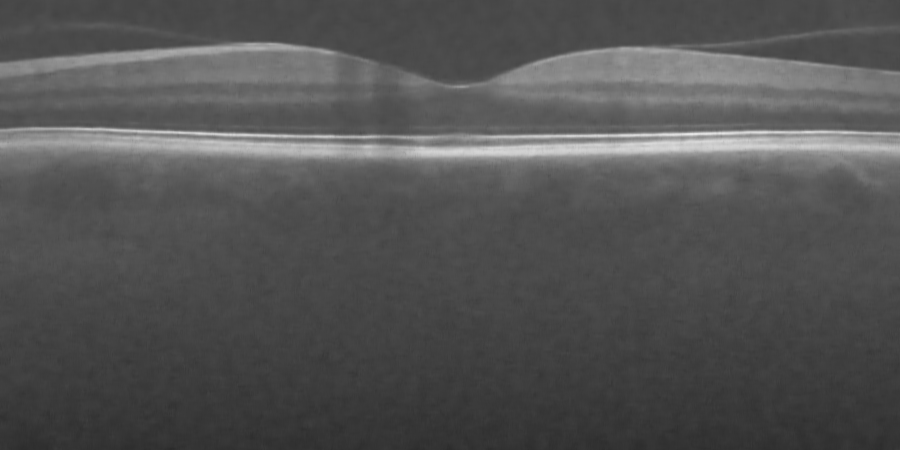} &
		\includegraphics[width=0.31\linewidth]{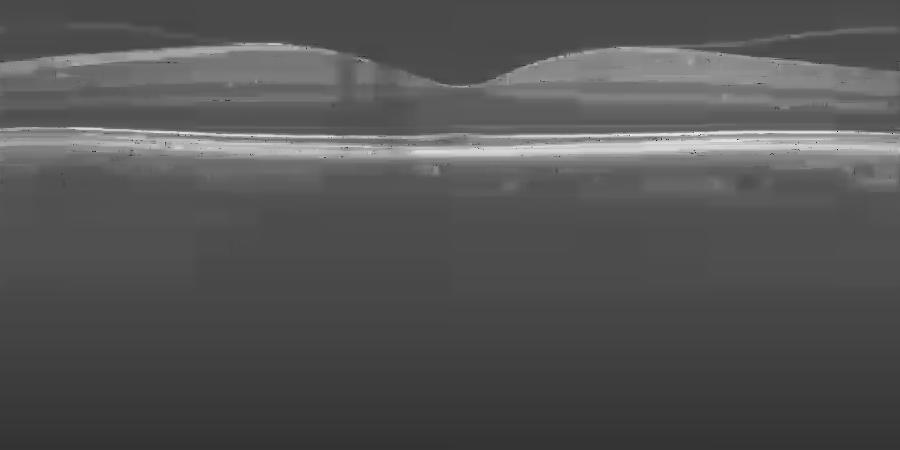} &
		\includegraphics[width=0.31\linewidth]{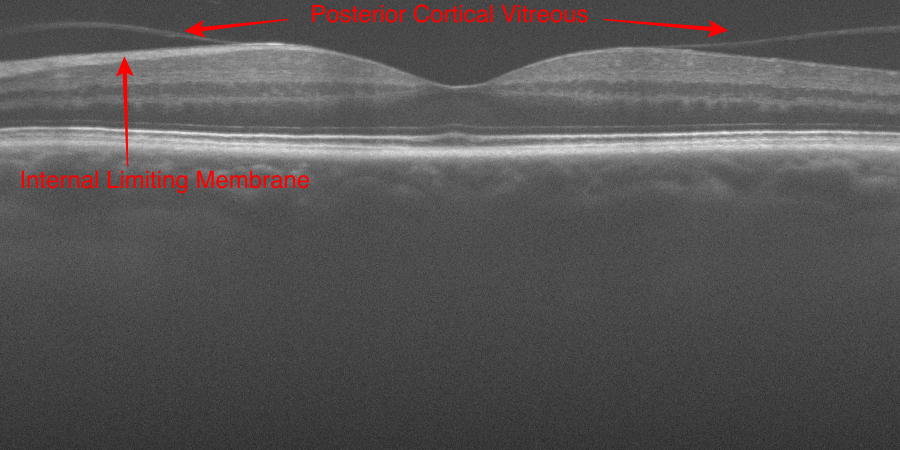} \\
		
		BPFA (11.7, 2.6) & WARP (10.2, 2.3) &  ``noiseless" 
	\end{tabular}
\end{figure*}

We further compare WARP with a study conducted in~\cite{Fang2017}, which considers another six method: BRFOE~\citep{weiss2007makes}, K-SVD~\citep{elad2006image}, PGPD~\citep{xu2015patch}, BM3D~\citep{dabov2007image}, MSBTD~\citep{fang2012sparsity}, and SSR~\citep{Fang2017}. These six methods have been applied to 18 foveal images from 18 subjects, using four slices nearby the original observation at various stages of their implementation. Although WARP does not require nearby information and can even process a 3D volume if such data exist, we apply WARP to the observation that averages the original observation and the four nearby slices only to make a fair comparison. In \ref{table:PSNR.18}, we adopt the mean of peak signal-to-noise ratio (PSNR) for all methods to align with \cite{Fang2017}, which is calculated as $-10 \log_{10}(\text{MSE})$ (noting that we rescale all observations and noiseless gray-scale images by 255). We can see that WARP gives the largest mean of PSNR, thus achieves excellent performances compared to a wide range of existing methods in this application setting. We choose the two subjects considered in~\ref{figure:OCT}, and plot the reconstructed images by WARP utilizing the nearby four slices in~\ref{figure:OCT.nearby}. It suggests that WARP even has an enhanced display compared to the ``noiseless" image, especially in the lower half of the image.

\begin{table}
	\caption{Mean PSNR for 18 foveal images reconstructed by BRFOE, K-SVD, PGPD, BM3D, MSBTD, SSR, and WARP. Results for the methods other than WARP are from~\cite{Fang2017}.}
	\label{table:PSNR.18} 
	\centering 
	\begin{tabular}{ccccccc}
		\toprule 
		BRFOE & K-SVD & PGPD & BM3D & MSBTD & SSR & WARP \\
		25.32 & 27.03 & 27.01 & 27.04 & 27.08 & 28.10 & 28.18 \\
		\bottomrule
	\end{tabular}
\end{table}

\begin{figure*}
	\caption{Reconstructed images using WARP based on the noisy observation and its four nearby slices. The two metrics following each method are the MSE ($\times 10^{-4})$ and MAE ($\times 10^{-2})$ respectively. The ``noiseless" reference is an registered and averaged image.}
	\label{figure:OCT.nearby}
	\begin{tabular}{c@{\hskip 0.1in}c}
		\includegraphics[width=0.45\linewidth]{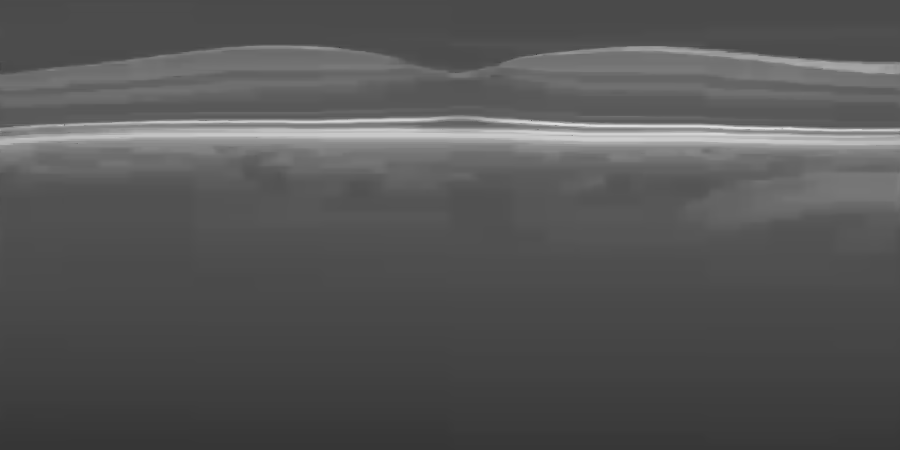} &
		\includegraphics[width=0.45\linewidth]{figure/OCT/\OCTalpha/average.png} \\
		WARP (6.6, 2.0) & ``noiseless" \\		
		\includegraphics[width=0.45\linewidth]{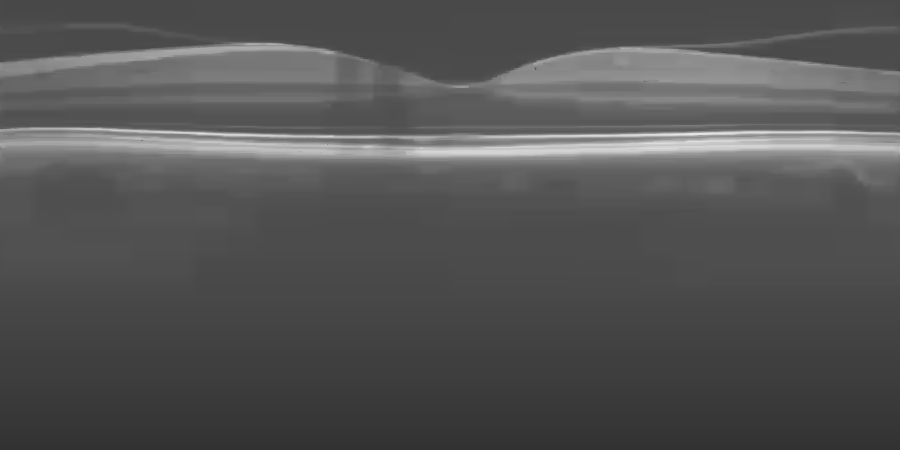} &
		\includegraphics[width=0.45\linewidth]{figure/OCT/\OCTbeta/average.png} \\
		WARP (10.0, 2.3) & ``noiseless"
	\end{tabular}
\end{figure*}

\section{Discussion\label{sec:dicussion}}
We have introduced the WARP framework that 
uses random recursive partitioning to induce a prior on the permutations of the index space, thereby 
achieving efficient inference on multi-dimensional functions
by converting it into a Bayesian model choice problem involving one-dimensional competitive generative models.
While our approach is Bayesian,  
one may consider other methods such as frequentist adaptive partitioning and shrinkage methods that incorporate the same idea. We do find satisfying the fully principled probabilistic inferential recipes that arise under our approach. 

The proposed framework WARP can be applied along with a wider range of Bayes wavelet regression models, including those that allow heterogeneous noise levels. If the error $\beps$ in Model~\eqref{eq:model} has general covariance matrix $\Sigma_{\epsilon}$, it often still makes sense to assume that the covariance of the error $\vect{u}$ in the wavelet domain, i.e. $W\Sigma_{\epsilon} W'$, is diagonal, due to the so-called whitening property of wavelet transforms discussed in~\cite{Johnstone1997}. In this case, let $\sigma_j^2 = \text{Var}(u_{j, k})$ for each $j$. Then one may estimate $\sigma_j^2$ using a robust estimator of the scale based on $\{w_{j, k}, 0 \leq k \leq 2^j - 1\}$ given a tree, for example, using the median absolute deviation of $\{w_{j, k}, 0 \leq k \leq 2^j - 1\}$ rescaled by 0.6745. Alternatively, one can adopt a hyperprior on location-based unknown variance $\sigma_j^2 \sim \text{IG}(\nu + 1, \nu \sigma_0^2)$, which is an inverse gamma prior with shape $\nu + 1$ and scale $\nu \sigma_0^2$ (thus the prior mean is $\sigma_0^2$). The hyperparameters $(\nu, \sigma_0^2)$ are either specified by users or estimated using data, for instance, we may estimate $\sigma_0^2$ by the median estimate based on the finest scale wavelet coefficients~\citep{donoho1995adapting} . 

Finally, while we introduce the WARP framework in the context of image denoising, we believe that the adaptive wavelet representation is applicable to a wide range of other tasks involving multi-dimensional signal processing.

\section*{Acknowledgments}
We are very grateful to an AE and three reviewers for providing extremely helpful comments and suggestions. We also thank Daniel Bourgeois for his help in porting our C++ code to R. Meng Li's research is partly supported by NSF grant DMS-2015569 and an ORAU Ralph E. Powe Junior Faculty Enhancement Award. Li Ma's research is partly supported by NSF grants DMS-1749789 and DMS-2013930. 

\section*{Supplementary Materials} 
Supplementary materials contain Proposition~\ref{lem:complexity.tree} and its proof; descriptions of WARP with local block shrinkage; details of the recursive message passing algorithm; proofs of all theorems;  
a sensitivity analysis for the proposed framework; plots of the 12 widely used test images used in Section~\ref{sec:DnCNN}; and comparison of WARP and selected methods using experiments of 3D image reconstruction. 

\bibliographystyle{IEEEtran}

\bibliography{WARP-abb}

\begin{thebibliography}{10}
\providecommand{\url}[1]{#1}
\csname url@samestyle\endcsname
\providecommand{\newblock}{\relax}
\providecommand{\bibinfo}[2]{#2}
\providecommand{\BIBentrySTDinterwordspacing}{\spaceskip=0pt\relax}
\providecommand{\BIBentryALTinterwordstretchfactor}{4}
\providecommand{\BIBentryALTinterwordspacing}{\spaceskip=\fontdimen2\font plus
\BIBentryALTinterwordstretchfactor\fontdimen3\font minus
  \fontdimen4\font\relax}
\providecommand{\BIBforeignlanguage}[2]{{%
\expandafter\ifx\csname l@#1\endcsname\relax
\typeout{** WARNING: IEEEtran.bst: No hyphenation pattern has been}%
\typeout{** loaded for the language `#1'. Using the pattern for}%
\typeout{** the default language instead.}%
\else
\language=\csname l@#1\endcsname
\fi
#2}}
\providecommand{\BIBdecl}{\relax}
\BIBdecl

\bibitem{alasil2010relationship}
T.~Alasil, P.~A. Keane, J.~F. Updike, L.~Dustin, Y.~Ouyang, A.~C. Walsh, and
  S.~R. Sadda, ``Relationship between optical coherence tomography retinal
  parameters and visual acuity in diabetic macular edema,''
  \emph{Ophthalmology}, vol. 117, no.~12, pp. 2379--2386, 2010.

\bibitem{bussel2013oct}
I.~I. Bussel, G.~Wollstein, and J.~S. Schuman, ``Oct for glaucoma diagnosis,
  screening and detection of glaucoma progression,'' \emph{British Journal of
  Ophthalmology}, pp. bjophthalmol--2013, 2013.

\bibitem{huang2014inner}
W.~C. Huang, A.~V. Cideciyan, A.~J. Roman, A.~Sumaroka, R.~Sheplock, S.~B.
  Schwartz, E.~M. Stone, and S.~G. Jacobson, ``Inner and outer retinal changes
  in retinal degenerations associated with abca4 mutations,''
  \emph{Investigative ophthalmology \& visual science}, vol.~55, no.~3, pp.
  1810--1822, 2014.

\bibitem{sun2014disorganization}
J.~K. Sun, M.~M. Lin, J.~Lammer, S.~Prager, R.~Sarangi, P.~S. Silva, and L.~P.
  Aiello, ``Disorganization of the retinal inner layers as a predictor of
  visual acuity in eyes with center-involved diabetic macular edema,''
  \emph{JAMA ophthalmology}, vol. 132, no.~11, pp. 1309--1316, 2014.

\bibitem{kafieh2015thickness}
R.~Kafieh, H.~Rabbani, F.~Hajizadeh, M.~D. Abramoff, and M.~Sonka, ``Thickness
  mapping of eleven retinal layers segmented using the diffusion maps method in
  normal eyes,'' \emph{Journal of ophthalmology}, vol. 2015, 2015.

\bibitem{oishi2018longitudinal}
A.~Oishi, P.~P. Fang, S.~Thiele, F.~G. Holz, and T.~U. Krohne, ``Longitudinal
  change of outer nuclear layer after retinal pigment epithelial tear secondary
  to age-related macular degeneration,'' \emph{Retina}, vol.~38, no.~7, pp.
  1331--1337, 2018.

\bibitem{Donoho1994}
D.~L. Donoho and I.~M. Johnstone, ``{Ideal spatial adapatation by wavelet
  shrinkage},'' \emph{Biometrika}, vol.~81, no.~3, pp. 425--455, 1994.

\bibitem{donoho1995adapting}
------, ``Adapting to unknown smoothness via wavelet shrinkage,'' \emph{J. Am.
  Statist. Ass.}, vol.~90, no. 432, pp. 1200--1224, 1995.

\bibitem{Mallat2008}
S.~Mallat, \emph{{A Wavelet Tour of Signal Processing: The Sparse Way}},
  3rd~ed.\hskip 1em plus 0.5em minus 0.4em\relax Academic press, 2008.

\bibitem{Abramovich+:98}
F.~Abramovich, T.~Sapatinas, and B.~W. Silverman, ``Wavelet thresholding via a
  {B}ayesian approach,'' \emph{J. R. Statist. Soc. B}, vol.~60, no.~4, pp.
  725--749, 1998.

\bibitem{crouse:1998}
M.~S. Crouse, R.~D. Nowak, and R.~G. Baraniuk, ``Wavelet-based statistical
  signal processing using hidden {M}arkov models,'' \emph{IEEE Transactions on
  Signal Processing}, vol.~46, no.~4, pp. 886--902, 1998.

\bibitem{Clyde+George:00}
M.~Clyde and E.~I. George, ``Flexible empirical {B}ayes estimation for
  wavelets,'' \emph{J. R. Statist. Soc. B}, vol.~62, no.~4, pp. 681--698, 2000.

\bibitem{Brown2001}
\BIBentryALTinterwordspacing
P.~J. Brown, T.~Fearn, and M.~Vannucci, ``Bayesian wavelet regression on curves
  with application to a spectroscopic calibration problem,'' \emph{J. Am.
  Statist. Ass.}, vol.~96, no. 454, pp. 398--408, jun 2001.
\BIBentrySTDinterwordspacing

\bibitem{Wil+Now:04}
R.~Willett and R.~Nowak, ``Fast multiresolution photon-limited image
  reconstruction,'' in \emph{IEEE International Symposium on Biomedical
  Imaging: Nano to Macro, 2004}, 2004, pp. 1192--1195.

\bibitem{Morris2006}
\BIBentryALTinterwordspacing
J.~S. Morris and R.~J. Carroll, ``{Wavelet-based functional mixed models},''
  \emph{J. R. Statist. Soc. B}, vol.~68, no.~2, pp. 179--199, apr 2006.
\BIBentrySTDinterwordspacing

\bibitem{Donoho:99}
D.~L. Donoho, ``Wedgelets: Nearly minimax estimation of edges,'' \emph{Ann.
  Statist.}, vol.~27, no.~3, pp. 859--897, 1999.

\bibitem{Jac+:11}
L.~Jacques, L.~Duval, C.~Chaux, and G.~Peyr{\'e}, ``A panorama on multiscale
  geometric representations, intertwining spatial, directional and frequency
  selectivity,'' \emph{Signal Processing}, vol.~91, no.~12, pp. 2699--2730,
  2011.

\bibitem{Ali+:13}
S.~T. Ali, J.-P. Antoine, and J.-P. Gazeau, \emph{{Coherent States, Wavelets
  and Their Generalizations}}, 2nd~ed.\hskip 1em plus 0.5em minus 0.4em\relax
  Springer, 2014.

\bibitem{raftery1995bayesian}
A.~E. Raftery, ``Bayesian model selection in social research,''
  \emph{Sociological Methodology}, vol.~25, pp. 111--163, 1995.

\bibitem{Volinsky1999}
\BIBentryALTinterwordspacing
J.~A. Hoeting, D.~Madigan, A.~E. Raftery, and C.~T. Volinsky, ``Bayesian model
  averaging: a tutorial,'' \emph{Statist. Sci.}, vol.~14, no.~4, pp. 382--417,
  11 1999.
\BIBentrySTDinterwordspacing

\bibitem{wongandma:2010}
W.~H. Wong and L.~Ma, ``Optional {P}{\'o}lya tree and {B}ayesian inference,''
  \emph{Ann. Statist.}, vol.~38, no.~3, pp. 1433--1459, 2010.

\bibitem{ma:2013}
L.~Ma, ``Adaptive testing of conditional association through recursive mixture
  modeling,'' \emph{J. Am. Statist. Ass.}, vol. 108, no. 504, pp. 1493--1505,
  2013.

\bibitem{Chipman1998a}
\BIBentryALTinterwordspacing
H.~A. Chipman, E.~I. George, and R.~E. McCulloch, ``{Bayesian CART model
  search},'' \emph{J. Am. Statist. Ass.}, vol.~93, no. 443, pp. 935--948, sep
  1998.
\BIBentrySTDinterwordspacing

\bibitem{Denison1998}
D.~G.~T. Denison, B.~K. Mallick, and A.~F.~M. Smith, ``{A Bayesian CART
  algorithm},'' \emph{Biometrika}, vol.~85, no.~2, pp. 363--377, 1998.

\bibitem{Johnstone+Silverman:05}
I.~M. Johnstone and B.~W. Silverman, ``Empirical {B}ayes selection of wavelet
  thresholds,'' \emph{Ann. Statist.}, vol.~33, no.~4, pp. 1700--1752, 2005.

\bibitem{Chipman1997}
\BIBentryALTinterwordspacing
H.~A. Chipman, E.~D. Kolaczyk, and R.~E. McCulloch, ``Adaptive {B}ayesian
  wavelet shrinkage,'' \emph{J. Am. Statist. Ass.}, vol.~92, no. 440, pp.
  1413--1421, dec 1997.
\BIBentrySTDinterwordspacing

\bibitem{ma&soriano:2018fanova}
\BIBentryALTinterwordspacing
L.~Ma and J.~Soriano, ``Efficient functional {ANOVA} through wavelet-domain
  {M}arkov groves,'' \emph{J. Am. Statist. Ass.}, 2017, to appear,
  DOI:10.1080/01621459.2017.1286241.
\BIBentrySTDinterwordspacing

\bibitem{Cai1999}
\BIBentryALTinterwordspacing
T.~T. Cai, ``{Adaptive wavelet estimation: A block thresholding and oracle
  inequality approach},'' \emph{Ann. Statist.}, vol.~27, no.~3, pp. 898--924,
  jun 1999.
\BIBentrySTDinterwordspacing

\bibitem{deng2009imagenet}
J.~Deng, W.~Dong, R.~Socher, L.-J. Li, K.~Li, and L.~Fei-Fei, ``Imagenet: A
  large-scale hierarchical image database,'' in \emph{2009 IEEE Conference on
  Computer Vision and Pattern Recognition}, 2009, pp. 248--255.

\bibitem{fryzlewicz2016shah}
P.~Fryzlewicz and C.~Timmermans, ``{SHAH: SHape-Adaptive Haar wavelets for
  image processing},'' \emph{Journal of Computational and Graphical
  Statistics}, vol.~25, no.~3, pp. 879--898, 2016.

\bibitem{polzehl2000adaptive}
J.~Polzehl and V.~G. Spokoiny, ``Adaptive weights smoothing with applications
  to image restoration,'' \emph{J. R. Statist. Soc. B}, vol.~62, no.~2, pp.
  335--354, 2000.

\bibitem{Li+Ghosal:14}
M.~Li and S.~Ghosal, ``Bayesian multiscale smoothing of {G}aussian noised
  images,'' \emph{Bayesian Analysis}, vol.~9, no.~3, pp. 733--758, 2014.

\bibitem{Cas+:04}
R.~Castro, R.~Willett, and R.~Nowak, ``Coarse-to-fine manifold learning,'' in
  \emph{IEEE International Conference on Acoustics, Speech, and Signal
  Processing, 2004 (ICASSP'04)}, 2004.

\bibitem{Zhou+:12}
M.~Zhou, H.~Chen, J.~Paisley, L.~Ren, L.~Li, Z.~Xing, D.~Dunson, G.~Sapiro, and
  L.~Carin, ``Nonparametric {B}ayesian dictionary learning for analysis of
  noisy and incomplete images,'' \emph{IEEE Transactions on Image Processing},
  vol.~21, no.~1, pp. 130--144, 2012.

\bibitem{Coi+Don:95}
R.~Coifman and D.~Donoho, ``\BIBforeignlanguage{English}{Translation-invariant
  de-noising},'' in \emph{\BIBforeignlanguage{English}{Wavelets and
  Statistics}}, ser. Lecture Notes in Statistics, A.~Antoniadis and
  G.~Oppenheim, Eds.\hskip 1em plus 0.5em minus 0.4em\relax Springer New York,
  1995, vol. 103, pp. 125--150.

\bibitem{Li+Ghosal:15}
M.~Li and S.~Ghosal, ``Fast translation invariant multiscale image denoising,''
  \emph{IEEE Transactions on Image Processing}, vol.~24, no.~12, pp.
  4876--4887, 2015.

\bibitem{Zhang2017}
K.~Zhang, W.~Zuo, Y.~Chen, D.~Meng, and L.~Zhang, ``{Beyond a Gaussian
  denoiser: Residual learning of deep CNN for image denoising},'' \emph{IEEE
  Transactions on Image Processing}, vol.~26, no.~7, pp. 3142--3155, 2017.

\bibitem{MartinFTM01}
D.~Martin, C.~Fowlkes, D.~Tal, and J.~Malik, ``A database of human segmented
  natural images and its application to evaluating segmentation algorithms and
  measuring ecological statistics,'' in \emph{Proc. 8th Int'l Conf. Computer
  Vision}, vol.~2, July 2001, pp. 416--423.

\bibitem{Liu2020CARP}
\BIBentryALTinterwordspacing
R.~Liu, M.~Li, and L.~Ma, ``{CARP: Compression Through Adaptive Recursive
  Partitioning for Multi-Dimensional Images},'' in \emph{2020 IEEE/CVF
  Conference on Computer Vision and Pattern Recognition (CVPR)}.\hskip 1em plus
  0.5em minus 0.4em\relax IEEE, jun 2020, pp. 14\,294--14\,302.
\BIBentrySTDinterwordspacing

\bibitem{huang1991optical}
D.~Huang, E.~A. Swanson, C.~P. Lin, J.~S. Schuman, W.~G. Stinson, W.~Chang,
  M.~R. Hee, T.~Flotte, K.~Gregory, C.~A. Puliafito \emph{et~al.}, ``Optical
  coherence tomography,'' \emph{Science}, vol. 254, no. 5035, pp. 1178--1181,
  1991.

\bibitem{grewal2013diagnosis}
D.~S. Grewal and A.~P. Tanna, ``Diagnosis of glaucoma and detection of glaucoma
  progression using spectral domain optical coherence tomography,''
  \emph{Current opinion in ophthalmology}, vol.~24, no.~2, pp. 150--161, 2013.

\bibitem{virgili2015optical}
G.~Virgili, F.~Menchini, G.~Casazza, R.~Hogg, R.~R. Das, X.~Wang, and
  M.~Michelessi, ``Optical coherence tomography (oct) for detection of macular
  oedema in patients with diabetic retinopathy,'' \emph{The Cochrane database
  of systematic reviews}, vol.~1, p. CD008081, 2015.

\bibitem{cuenca2018cellular}
N.~Cuenca, I.~Ortu{\~n}o-Lizar{\'a}n, and I.~Pinilla, ``Cellular
  characterization of oct and outer retinal bands using specific
  immunohistochemistry markers and clinical implications,''
  \emph{Ophthalmology}, vol. 125, no.~3, pp. 407--422, 2018.

\bibitem{keane2012evaluation}
P.~A. Keane, P.~J. Patel, S.~Liakopoulos, F.~M. Heussen, S.~R. Sadda, and
  A.~Tufail, ``Evaluation of age-related macular degeneration with optical
  coherence tomography,'' \emph{Survey of ophthalmology}, vol.~57, no.~5, pp.
  389--414, 2012.

\bibitem{shi2015automated}
F.~Shi, X.~Chen, H.~Zhao, W.~Zhu, D.~Xiang, E.~Gao, M.~Sonka, and H.~Chen,
  ``Automated 3-d retinal layer segmentation of macular optical coherence
  tomography images with serous pigment epithelial detachments.'' \emph{IEEE
  Trans. Med. Imaging}, vol.~34, no.~2, pp. 441--452, 2015.

\bibitem{Fang2017}
\BIBentryALTinterwordspacing
L.~Fang, S.~Li, D.~Cunefare, and S.~Farsiu, ``{Segmentation Based Sparse
  Reconstruction of Optical Coherence Tomography Images},'' \emph{IEEE
  Transactions on Medical Imaging}, vol.~36, no.~2, pp. 407--421, feb 2017.
\BIBentrySTDinterwordspacing

\bibitem{duker2013international}
J.~S. Duker, P.~K. Kaiser, S.~Binder, M.~D. De~Smet, A.~Gaudric, E.~Reichel,
  S.~R. Sadda, J.~Sebag, R.~F. Spaide, and P.~Stalmans, ``The international
  vitreomacular traction study group classification of vitreomacular adhesion,
  traction, and macular hole,'' \emph{Ophthalmology}, vol. 120, no.~12, pp.
  2611--2619, 2013.

\bibitem{weiss2007makes}
Y.~Weiss and W.~T. Freeman, ``What makes a good model of natural images?'' in
  \emph{2007 IEEE Conference on Computer Vision and Pattern Recognition}.\hskip
  1em plus 0.5em minus 0.4em\relax IEEE, 2007, pp. 1--8.

\bibitem{elad2006image}
M.~Elad and M.~Aharon, ``Image denoising via sparse and redundant
  representations over learned dictionaries,'' \emph{IEEE Transactions on Image
  processing}, vol.~15, no.~12, pp. 3736--3745, 2006.

\bibitem{xu2015patch}
J.~Xu, L.~Zhang, W.~Zuo, D.~Zhang, and X.~Feng, ``Patch group based nonlocal
  self-similarity prior learning for image denoising,'' in \emph{Proceedings of
  the IEEE international conference on computer vision}, 2015, pp. 244--252.

\bibitem{dabov2007image}
K.~Dabov, A.~Foi, V.~Katkovnik, and K.~Egiazarian, ``Image denoising by sparse
  3-d transform-domain collaborative filtering,'' \emph{IEEE Transactions on
  image processing}, vol.~16, no.~8, pp. 2080--2095, 2007.

\bibitem{fang2012sparsity}
L.~Fang, S.~Li, Q.~Nie, J.~A. Izatt, C.~A. Toth, and S.~Farsiu, ``Sparsity
  based denoising of spectral domain optical coherence tomography images,''
  \emph{Biomedical optics express}, vol.~3, no.~5, pp. 927--942, 2012.

\bibitem{Johnstone1997}
\BIBentryALTinterwordspacing
I.~M. Johnstone and B.~W. Silverman, ``Wavelet threshold estimators for data
  with correlated noise,'' \emph{J. R. Statist. Soc. B}, vol.~59, no.~2, pp.
  319--351, 1997.
\BIBentrySTDinterwordspacing

\bibitem{Aho+Slo:73}
A.~V. Aho and N.~J. Sloane, ``Some doubly exponential sequences,''
  \emph{Fibonacci Quart}, vol.~11, no.~4, pp. 429--437, 1973.

\bibitem{Muk+Qiu:11}
P.~Mukherjee and P.~Qiu, ``3-d image denoising by local smoothing and
  nonparametric regression,'' \emph{Technometrics}, vol.~53, no.~2, pp.
  196--208, 2011.

\bibitem{Per+Mal:90}
P.~Perona and J.~Malik, ``Scale-space and edge detection using anisotropic
  diffusion,'' \emph{IEEE Transactions on Pattern Analysis and Machine
  Intelligence}, vol.~12, no.~7, pp. 629--639, 1990.

\bibitem{Rud+:92}
L.~Rudin, S.~Osher, and E.~Fatemi, ``Nonlinear total variation based noise
  removal algorithms,'' \emph{Physica D: Nonlinear Phenomena}, vol.~60, no.~1,
  pp. 259--268, 1992.

\bibitem{Cou+:08}
P.~Coup{\'e}, P.~Yger, S.~Prima, P.~Hellier, C.~Kervrann, and C.~Barillot, ``An
  optimized blockwise nonlocal means denoising filter for 3-d magnetic
  resonance images,'' \emph{IEEE Transactions on Medical Imaging}, vol.~27,
  no.~4, pp. 425--441, 2008.

\end{thebibliography}

\newpage 
\clearpage
\setcounter{page}{1}
\renewcommand*{\thepage}{S-\arabic{page}}
\begin{NoHyper}
	\begin{center} 
		\Large
		Supplementary Materials for ``\papertitle" 
	\end{center} 
	
	\setcounter{section}{0}
	\renewcommand{\thesection}{\Alph{section}} 
	\renewcommand{\thesubsection}{\thesection.\Alph{subsection}}
	
	Supplementary materials contain (A) Proposition~\ref{lem:complexity.tree} and its proof, (B) descriptions of WARP with local block shrinkage, (C) details of the recursive message passing algorithm, (D) proofs of all theorems, (E) a sensitivity analysis for the proposed framework, (F) plots of the 12 widely used test images used in Section~\ref{sec:DnCNN}, and (G) comparison of WARP and selected methods using experiments of 3D image reconstruction. 
	
	\section{Cardinality of the space of RDPs}
	\begin{prop}
		\label{lem:complexity.tree} 
		The log cardinality of the tree space induced by RDPs is $O(n)$ when $m = 2$. 
	\end{prop} 
	
	\begin{proof}[Proof of Proposition~\ref{lem:complexity.tree}]
		Let $c(a, b)$ be the cardinality of the tree space induced by RDPs for an $2^a$ by $2^b$ image. We can obtain the following recursive formula 
		\begin{equation}
		c(a, b) = 
		\begin{cases}
		c^2(a - 1, b) + c^2(a, b - 1), & \text{if } a, b \geq 1 \\
		1 & \text{if }  a = 0 \text{ or } b = 0. 
		\end{cases}
		\end{equation}
		We assert that there exist two constants $(k_1, k_2)$ such that $k_2 \geq k_1 > 1$ and 
		\begin{equation}
		c(a, b) \in \left[\frac{1}{2} k_1^{2^{a + b}}, \frac{1}{2}k_2^{2^{a + b}}\right], 
		\end{equation}
		for any $a \geq 1$ and $b \geq 1$. 
		
		First consider $a = 1$ and $b \geq 1$. We have $
		c(1, b) = c^2(1, b - 1) + 1
		$
		when $b \geq 1$ and $c(1, 0) = 1$ when $b = 0$. The quantity $c(1, b)$ is actually the number of ``strongly" binary trees of height $\leq b$, which possesses an analytical form 
		\begin{equation}
		c(1, b) = \lfloor k^{2^b} \rfloor, 
		\end{equation}
		according to \cite{Aho+Slo:73}, where 
		\begin{equation}
		k = \exp\left\{\sum_{j = 0}^\infty 2^{-j - 1} \log(1 + c^{-2}(1, j))\right\} \approx 1.503. 
		\end{equation}
		Letting $k_1 = \sqrt{k}$ and $k_2 = k$ and noting $k^{2^b} \geq 2$ for all $b \geq 1$, we  obtain that
		\begin{equation}
		\frac{1}{2} k_1^{2^{1 + b}} = \frac{1}{2} k^{2^{b}}  \leq  k^{2^{b}}  - 1 \leq \lfloor k^{2^b} \rfloor \leq k^{2^{b}}  \leq \frac{1}{2}k^{2^{1 + b}}, 
		\end{equation}
		for all $b \geq 1$. Therefore, the assertion holds for all $a = 1$ and $b \geq 1$.  Since $c(a, b) = c(b, a)$, the assertion also holds for all $a \geq 1$ and $b = 1$. 
		
		For any $a \geq 1$ and $b \geq 1$, it is easy to verify that if the assertion holds for $(a, b - 1)$ and $(a - 1, b)$, then it holds for $(a, b)$.  We then complete the proof by  induction. 
	\end{proof} 
	
	\section{WARP with local block shrinkage\label{sec:RRDP.OP}}
	
	Traditional wavelet analysis is done 
	by fixing the maximum depth of the wavelet tree at $J$. That is, one partitions the index space all the way down to the finest level of ``atomic'' blocks. In most practical problems, once the blocks are small enough, the function value within the block becomes essentially constant with respect to the noise level, and so further division within such homogeneous blocks will be wasteful and will reduce statistical efficiency. For example, in~\ref{fig:RDP} the partition in the upper left block (Level 3) along with its descendants is not necessary.
	Thus it is often desirable to incorporate adaptivity in the depth of the wavelet tree and allow it to be terminated earlier than reaching level $J$. In practice the optimal maximum depth varies across $\om$. 
	For example, some parts of an image may contain many interesting details, while the rest do not---e.g., an image of a painting hung on a gray wall. A high resolution will be needed to capture the details in the painting, but would be unnecessary and introduce additional variability in the estimation for the wall. 
	
	This consideration is closely related to the idea of adaptive block shrinkage \citep{Cai1999} in the frequentist wavelet regression analysis. Once there is little evidence for any interesting structure within a subset of the index space, then the function value within that subset can be shrunk to a constant. That is, the wavelet tree is ``pruned'' there. Next we show that such pruning can be achieved in a hierarchical modeling manner, and the resulting Bayesian wavelet regression model is again compatible with our WARP framework.

	To achieve such pruning, we introduce another set of latent variables $\R=\{R_{j,k}: j=0,1,\ldots, J-1, k=0,1,\ldots,2^{j}-1\}$, where $R_{j,k}=1$ indicates that the tree is pruned at node $(j,k)$. Next we describe a generative prior on $\R$ that will blend well with the WARP framework. To start, let $
	R_{0,0} \ind {\rm Bern}(\eta_{0,0})$
	and for all $j\geq 1$, and
	\[
	R_{j,k}\,|\,R_{j-1,\lfloor k/2 \rfloor} \ind \begin{cases}
		{\rm Bern}(\eta_{j,k}) & \text{if $R_{j-1,\lfloor k/2 \rfloor}=0$}\\
		{\rm Bern}(1) & \text{if $R_{j-1,\lfloor k/2 \rfloor}=1$}.
	\end{cases}
	\]
	That is, if a node's parent has been pruned, then its children are also pruned by construction.
	From now on, we shall refer to this prior model on $\R$ as an {\em optional pruning} (OP) model \citep{ma:2013}, which is specified by a set of {\em pruning probabilities} $\eta_{j,k}\in [0,1]$. We write $\R\sim {\rm OP}(\bm{\eta})$. 
	
	Given $\R$, we can modify our prior on $\S$ to reflect the effect of pruning. For example, instead of an independent prior on $S$, we can now generate them as follows
	\[
	S_{j,k}\,|\,\R \ind \begin{cases}
		{\rm Bern}(\rho_{j, k}) & \text{if $R_{j,k}=0$}\\
		{\rm Bern}(0) & \text{if $R_{j,k}=1$.}
	\end{cases}
	\]
	That is, if the node has not been pruned, then we generate $S_{j,k}$ from the independent Bernoulli as in the standard spike-and-slab setup, but if the node has been pruned, then by construction, we must have $S_{j,k}=0$ due to pruning.
	
	It is often reasonable to specify the prior shrinkage and pruning probabilities as functions of the level in the RDP. That is, $\rho_{j, k} = \rho_j$ and $\eta_{j, k} = \eta_j$ for all $k$. In the node-specific notation, $\rho(A)=\rho_j$ and $\eta(A)=\eta_j$ for all $j$th node $A\in \A$. In this case, one can show that this joint model on $(\S,\R)$ is equivalent to a Markov tree model with three states defined in terms of the combinations of $(S_{j,k},R_{j,k})=(1,0)$ , (0,0), or (0,1), and with the corresponding transition matrix for $S_{j,k}$ given by
	\[
	\brho_j = 
	\begin{bmatrix}
		\rho_j (1-\eta_j) & (1-\rho_j)(1-\eta_j) & \eta_j\\
		\rho_j (1-\eta_j) & (1-\rho_j)(1-\eta_j) & \eta_j\\
		0 & 0 & 1
	\end{bmatrix}.
	\] 
	This allows us to derive the posterior from Theorem~\ref{thm:latent_state}, and carry out inference accordingly. Specifically, for each $A \in \A$, let $p_0(A)$ be the marginal likelihood contributed from the wavelet coefficients in $A$ and its descendants if $A$ is pruned, i.e., 
	$$p_0(A)=\frac{1}{(\sqrt{2\pi\sigma^2})^{|A|-1}} \exp\left\{-\frac{\sum_{x\in A}(y(x)-\bar{y}(A))^2}{2\sigma^2}\right\}, $$ where $\bar{y}{(A)} = \sum_{x \in A} y(x)/|A|$. If $A \in \T$, the following maps are directly available from Theorem~\ref{thm:latent_state}:
	\begin{itemize}
		\item The marginal likelihood contribution from the data within node $A$ if $A$ is divided in dimension $d$: 
		$$M_{d}(A)=\rho(A) M_{d}^{(1)}(A) + (1-\rho(A)) M_{d}^{(0)}(A);$$ 
		\item The posterior spike probability $\tilde{\rho}_{d}$ of $A$ if $A$ is divided in dimension $d$:  $$\tilde{\rho}_{d}(A)=\rho(A)M^{(1)}_{d}(A)/M_{d}(A);$$ 
		\item The marginal likelihood from data on $A$ and its descendants: $\Psi(A)=(1-\eta(A))\sum_{d\in \D(A)} \lambda_d(A) M_{d}(A) \Psi(A_{l}^{(d)})\Psi(A_{r}^{(d)}) + \eta(A) p_{0}(A)$ \text{if $A$ is non-atomic;} $\Psi(A)=1$ \text{if $A$ is atomic.} 
		\item The posterior probability of pruning $A$: $$\tilde{\eta}(A)=\eta(A)p_{0}(A)/\Psi(A);$$ 
		\item  The posterior probability for $A$ to be divided in dimension $d$ given $A$ is not pruned: 
		$$\tilde{\lambda}_{d}(A)=(1 - \eta(A))\frac{\lambda_{d}(A) M_{d}(A)\Psi(A^{(d)}_{l})\Psi(A^{(d)}_{r})}{\Psi(A)-\eta(A)p_{0}(A)}.$$
	\end{itemize}  
	
	\section{Recursive message passing algorithm} \label{sec:message.passing} 
	For the Haar basis, the posterior mean ${\rm E}(\vect{f}|\by)$ can be evaluated analytically through recursive message passing without any Monte Carlo sampling for Bayesian wavelet regression models that adopt the spike-and-slab setup along with optional pruning of the wavelet tree, which contains the models without optional pruning as special cases with zero pruning probabilities. We describe the strategy next and will use it to compute ${\rm E}(\vect{f}|\by)$ in our numerical examples. 
	
	For each $A\in \A$, let $c(A)$ be the scale (father wavelet) coefficient on $A$ if $A\in\T$, and let $\varphi(A)= {\rm E}(c(A) \I_{\{A\in\T\}} \,|\, \by)$.  Note that ${\rm E}(\vect{f}\,|\,\by)$ is given by $\varphi(A)$ for all atomic $A$.  To compute the mapping $\varphi$, we introduce two auxiliary mappings
	$\psi_0(A) = {\rm P}(A\in\T,R(A)=0\,|\,\by)$ and $\varphi_0(A) = {\rm E}(c(A) \I_{\{A\in\T, R(A)=0\}} \,|\, \by). $
	Let $\bar{A}^{(d)}$ denote the parent of $A$ in $\T$ if $A$ is a child node after dividing its parent in the $d$th dimension, and let $\mathcal{P}(A)\subset\{1,2,\ldots,m\} $ be the collection of dimensions of $A$ that do not have full support $[0,2^{J_i} - 1]$, i.e., those that have been partitioned at least once in previous levels.
uting the tri-variate mapping $(\phi_0, \varphi_0, \varphi):\A\rightarrow \real^3$.
	
	\begin{thm}
		\label{thm:recursive.map}
		To initiate the recursion, for $A = \Omega$, we let
		$\psi_0(A) = 1 - \tilde{\eta}(A), \varphi_0(A) = (1 - \tilde{\eta}(A)) |A|/\sqrt{n}$, and $\varphi(A) =  |A|/\sqrt{n}$.
		Suppose we have evaluated these mappings up to level $j-1$, for level $j = 1, \ldots, J$, we have
		\begin{align}
			\psi_0(A) & = \sum_{d\in\mathcal{P}(A)} \psi_0(\bar{A}^{(d)}) \tilde{\lambda}_{d}(\bar{A}^{(d)})(1-\tilde{\eta}(A)); \\
			\varphi_0(A) & = (1 - \tilde{\eta}(A)) \cdot \sum_{d\in\mathcal{P}(A)} \frac{ \tilde{\lambda}_{d}(\bar{A}^{(d)})  }{\sqrt{2}}\bigg[\varphi_0(\bar{A}^{(d)}) -   \\
			& \hspace*{-0.2in}\tilde{\rho}_{d}(\bar{A}^{(d)}) \mu_1(w_{d}(\bar{A}^{(d)})) \psi_0(\bar{A}^{(d)})(-1)^{\I(\text{$A$ is the left child of $\bar{A}^{(d)}$})} \bigg];  \\
			\varphi(A) & = \frac{\varphi_0(A)}{1 - \tilde{\eta}(A)} \hspace*{-0.03in}+\hspace*{-0.03in} \frac{1}{\sqrt{2}} \hspace*{-0.06in}\sum_{ d\in\mathcal{P}(A)}\hspace*{-0.1in} \{\varphi(\bar{A}^{(d)}) - \varphi_0(\bar{A}^{(d)}) \}\lambda_{d}(\bar{A}^{(d)}).
		\end{align}
	\end{thm}
	\noindent Remark: Note that this recursion is top-down (from low to high resolutions), whereas that for computing $\Phi$ is bottom-up (from high to low resolutions). The two-directional recursion shares the spirit of the forward-backward algorithm for HMMs.
	
	Once we have computed the mapping $(\varphi_0,\psi_0,\varphi): \A\rightarrow \real^3$, the posterior mean ${\rm E}(\vect{f} \,|\,\by)$ is then given by $\varphi$ applied on the atomic nodes. Note that this theorem applies to the special case with no pruning as well.

	\section{Proofs of Theorems}
	\label{sec:proofs} 
	
	\begin{proof}[Proof of Theorem~\ref{thm:independent_shrinkage}]
		Because Theorem~\ref{thm:independent_shrinkage} can be considered a special case with a single latent state, its proof follows immediately from the latter theorem, which we prove below.
	\end{proof}
	
	\begin{proof}[Proof of Theorem~\ref{thm:latent_state}]
		First we verify that the mapping $\Phi_{s}(A)$ is the marginal likelihood contributed from data with locations in $A$, given that $A\in \T$ and that the latent state variable associated with the parent of $A$ in $\T$ is $s$. We show this by induction. First note that if $A$ is atomic, then
		\[
		\Phi_s(A) = {\rm P}(\by(A) \,|\,A\in\T,S(A_p)=s)=1
		\]
		by design as there are no wavelet coefficients associated with atomic nodes. Now, suppose we have shown that $\Phi_s(A)= {\rm P}(\by(A) \,|\,A\in\T,S(A_p)=s)$ for all $A$ with level higher than $j$. Then if $A$ is of level $j$, it follows that 
		\begin{align}
		&{\rm P}(\by(A) \,|\,A\in\T,S(A_p)=s)\\
		=& \sum_{s'}\sum_{d}{\rm P}(\by(A) \,|\,A\in \T,S(A)=s',S(A_p)=s,D(A)=d) \\
		&\hspace{5em} \times {\rm P}(S(A)=s'\,|\,A\in\T,S(A_p)=s) \\
		& \hspace{5em}\times{\rm P}(D(A)=d\,|\,A\in\T,S(A_p)=s)\\
		=& \sum_{s'}\rho_j(s,s')\sum_{d \in \mathcal{D}(A)} \lambda_{d} M^{(s')}_d(A) \Psi_{s'}(A_l^{d}) \Psi_{s'}(A_{r}^{(d)}),
		\end{align}
		which leads to the definition of $\Phi_s(A)$ in Theorem~\ref{thm:latent_state}.  
		
		Next let us derive the joint marginal posterior of $(\T,\S)$. Note that
		\begin{align}
		& {\rm P}(S_{j,k}=s'\,|\,S_{j-1,\lfloor k/2 \rfloor} = s,\T^{(j)},\by) \\ = & \frac{{\rm P}(S_{j,k}=s',S_{j-1,\lfloor k/2 \rfloor} = s,\by(A)\,|\,\T^{(j)})}{{\rm P}(S_{j-1,\lfloor k/2 \rfloor} = s,\by(A)\,|\,\T^{(j)})}.
		\end{align}
		Now we have 
		\begin{align}
		& {\rm P}(S_{j,k}=s',D_{j,k}=d,\by(A)\,|\,\T^{(j)},S_{j-1,\lfloor k/2 \rfloor} = s) \\ = & \rho_{j}(s,s') \lambda_{d}(A) M_{d}^{(s')}(A) \Phi_{s'}(A^{(d)}_{l})\Phi_{s'}(A^{(d)}_{r}), 
		\end{align}
		which leads to 
		\begin{align}
		& {\rm P}(S_{j,k}=s',\by(A)\,|\,\T^{(j)},S_{j-1,\lfloor k/2 \rfloor} = s) \\ =&\rho_{j}(s,s') \sum_{d}\lambda_{d}(A) M_{d}^{(s')}(A) \Phi_{s'}(A^{(d)}_{l})\Phi_{s'}(A^{(d)}_{r})
		\end{align}
		and furthermore, 
		\begin{align}
		& {\rm P}(S_{j,k}=s'\,|\,S_{j-1,\lfloor k/2 \rfloor} = s,\T^{(j)},\by) \\ = & \frac{\rho_{j}(s,s') \sum_{d}\lambda_{d}(A) M_{d}^{(s')}(A) \Phi_{s'}(A^{(d)}_{l})\Phi_{s'}(A^{(d)}_{r})}{\sum_{s''}\rho_{j}(s,s'') \sum_{d}\lambda_{d}(A) M_{d}^{(s'')}(A) \Phi_{s''}(A^{(d)}_{l})\Phi_{s''}(A^{(d)}_{r})}, 
		\end{align}
		where the denominator is just $\Phi_s(A)$.
		
		Finally,
		\begin{align}
		& {\rm P}(D_{j,k}=d\,|\,S_{j,k}=s',\T^{(j)},\by) \\ = & \frac{{\rm P}(D_{j,k}=d,\by(A)\,|\,S_{j,k}=s',\T^{(j)})}{{\rm P}(\by(A)\,|\,S_{j,k}=s',\T^{(j)})}\\
		= & \frac{\lambda_{d}(A) M_{d}^{(s')}(A) \Phi_{s'}(A^{(d)}_{l})\Phi_{s'}(A^{(d)}_{r})}{\sum_{d'} \lambda_{d'}(A) M_{d'}^{(s')}(A) \Phi_{s'}(A^{(d')}_{l})\Phi_{s'}(A^{(d)}_{r})}.
		\end{align}
		This completes the proof.
	\end{proof}

	\begin{proof}[Proof of Theorem~\ref{thm:recursive.map}]
		We first obtain the recursive recipe for computing the maps $(\psi_0, \varphi_0)$ following Theorem~\ref{thm:independent_shrinkage}:
		\begin{align}
		& \psi_0(A) \\ = & \sum_{d\in \mathcal{P}(A)} {\rm P}(\bar{A}^{(d)}\in\T,R(\bar{A}^{(d)})=0\,|\,\by) \tilde{\lambda}_{d}(\bar{A}^{(d)})(1-\tilde{\eta}(A)) \\
		 = & \sum_{d\in\mathcal{P}(A)} \psi_0(\bar{A}^{(d)}) \tilde{\lambda}_{d}(\bar{A}^{(d)})(1-\tilde{\eta}(A)),
		\end{align}
		and
		\begin{align}
		& \varphi_0(A) = {\rm E}\left(c(A)\I_{\{A\in\T,R(A)=0\}}\,|\,\by\right) \\
		=  & \sum_{ d\in\mathcal{P}(A)} {\rm E}\left(c(A)\I_{\{\bar{A}^{(d)}\in\T,D(\bar{A}^{(d)})=d, R(A) = 0\}}\,|\,\by\right) \\
		=&\sum_{d\in\mathcal{P}(A)}{\rm E}\left(c(A)\,|\,\bar{A}^{(d)}\in\T,D(\bar{A}^{(d)})=d,R(A)=0,\by  \right)\\
		& \times {\rm P}\left( \bar{A}^{(d)}\in\T,D(\bar{A}^{(d)})=d,R(A)=0\,|\,\by\right) \\
		=& (1 - \tilde{\eta}(A)) \sum_{d\in\mathcal{P}(A)} \frac{\tilde{\lambda}_{d}(\bar{A}^{(d)}) }{\sqrt{2}}\bigg[\varphi_0(\bar{A}^{(d)}) - \\ & \hspace*{-0.5em}{\tilde{\rho}_{d}(\bar{A}^{(d)}) \mu_1(w_{d}(\bar{A}^{(d)}))\psi_0(\bar{A}^{(d)})}\cdot (-1)^{\I(\text{$A$ is the left child of $\bar{A}^{(d)}$})}\bigg]. \\
		\label{eq:varphi.0}
		\end{align}
		
		We next derive the recursive formula for $\varphi(A)$. Let $\varphi_1(A) = {\rm E}(c(A) \I_{\{A\in\T, R(A)=1\}} \,|\, \by)$, then we have $\varphi(A) = {\rm E}(c(A) \I_{\{A\in\T\}} \,|\, \by) = \varphi_0(A) + \varphi_1(A) $.
		Note that
		\begin{equation}
		\label{eq:dumm3}
		\varphi(A) = \sum_{ d\in\mathcal{P}(A)} {\rm E}\left(c(A)\I_{\{\bar{A}^{(d)}\in\T,D(\bar{A}^{(d)})=d\}}\,|\,\by\right),
		\end{equation}
		and for each $ d\in\mathcal{P}(A)$, we have
		\begin{align}
		& {\rm E}\left(c(A)\I_{\{\bar{A}^{(d)}\in\T,D(\bar{A}^{(d)})=d\}}\,|\,\by\right) \\
		= & \sum_{r = 0, 1}{\rm E}\left(c(A)\I_{\{\bar{A}^{(d)}\in\T,D(\bar{A}^{(d)})=d, R(\bar{A}^{(d)}) = r\}}\,|\,\by\right) \label{eq:dummy2}\\
		=&  \sum_{r = 0, 1}
		{\rm E}\left(c(A) \,|\, \bar{A}^{(d)}\in\T,D(\bar{A}^{(d)})=d, R(\bar{A}^{(d)}) = r, \by\right) \\
		& \qquad   
		\times {\rm P}(\bar{A}^{(d)}\in\T,D(\bar{A}^{(d)})=d, R(\bar{A}^{(d)}) = r \,|\, \by). \label{eq:dummy1}
		\end{align}
		For the second term in~\eqref{eq:dummy1}, we have
		\begin{align}
		& {\rm P}(\bar{A}^{(d)}\in\T,D(\bar{A}^{(d)})=d, R(\bar{A}^{(d)}) = r \,|\, \by) \\
		=&{\rm P}(D(\bar{A}^{(d)})=d \,|\, \bar{A}^{(d)}\in\T, R(\bar{A}^{(d)}) = r, \by) \\
		& \hspace{3em} \times {\rm P}(\bar{A}^{(d)}\in\T, R(\bar{A}^{(d)}) = r \,|\, \by) \\
		=& \tilde{\lambda}_{d}(\bar{A}^{(d)})^{1-r} \lambda_{d}(\bar{A}^{(d)})^{r} \psi_{r}(\bar{A}^{(d)}). 
		\end{align}
		For the first term in~\eqref{eq:dummy1},  it is easy to check that
		\begin{align}
		&{\rm E}\left(c(A)\,|\,\bar{A}^{(d)}\in\T,D(\bar{A}^{(d)})=d,R(\bar{A}^{(d)})=r,\by  \right)\\
		=&\begin{cases}
		\frac{1}{\sqrt{2}}\bigg[\frac{\varphi_0(\bar{A}^{(d)})}{\psi_0(\bar{A}^{(d)})} -  \tilde{\rho}_{d}(\bar{A}^{(d)}) \mu_1(w_{d}(\bar{A}^{(d)})) \\ \hspace{5em} \times (-1)^{\I(\text{$A$ is the left child of $\bar{A}^{(d)}$})}\bigg] & \text{if $r=0$}\\
		\frac{1}{\sqrt{2}} \varphi_1(\bar{A}^{(d)})/\psi_1(\bar{A}^{(d)})& \text{if $r=1$},
		\end{cases}
		\end{align}
		where we use the independence between $c(A)$ and $D(A)$ given $A \in \T$.
		Plugging the two terms into~\eqref{eq:dummy1} , we obtain that
		\begin{align}
		& {\rm E}\left(c(A)\I_{\{\bar{A}^{(d)}\in\T,D(\bar{A}^{(d)})=d\}}\,|\,\by\right)\\
		=&\frac{1}{\sqrt{2}}\bigg[\varphi_0(\bar{A}^{(d)})  - \tilde{\rho}_{d}(\bar{A}^{(d)}) w_{d}(\bar{A}^{(d)})/(1+\tau_{j-1}^{-1}) \\ & \hspace{3em} \times (-1)^{\I(\text{$A$ is the left child of $\bar{A}^{(d)}$})}\cdot \psi_{0}(\bar{A}^{(d)})\bigg] \tilde{\lambda}_{d}(\bar{A}^{(d)}) \\
		& + \frac{1}{\sqrt{2}} \varphi_1(\bar{A}^{(d)}) \lambda_{d}(\bar{A}^{(d)}). \label{eq:dummy4}
		\end{align}
		
		Combining the result in~\eqref{eq:dumm3} and~\eqref{eq:dummy4}, and comparing it with $\varphi_0(A)$ in~\eqref{eq:varphi.0}, we obtain that
		\begin{equation}
		\varphi(A) = \varphi_0(A)/(1 - \tilde{\eta}(A)) + \frac{1}{\sqrt{2}} \sum_{ d\in\mathcal{P}(A)} \varphi_1(\bar{A}^{(d)}) \lambda_{d}(\bar{A}^{(d)}),
		\end{equation}
		which concludes the proof by plugging in $\varphi_1(\cdot) = \varphi(\cdot) - \varphi_0(\cdot)$.
	\end{proof}

	\section{Sensitivity analysis\label{sec:sensitivity.gamma} } 
	In this section, we conduct a sensitivity analysis for the proposed WARP framework at various choices of hyperparameters. 
	
	We first implement the method of ``WARP-full" which chooses $\bphi$ by a full optimization of the marginal likelihood using two simulated images $(f_1, f_2)$ explicitly given in~\ref{sec:experiment.3D}. Recall that the row of WARP selects $\bphi$ at a limited set of grid points. ~\ref{table:sensitivity.3D} shows that the MSEs of WARP-full are almost identical to the row of WARP. This observation is consistent across many scenarios we have tested. Therefore, the method of WARP seems robust in terms of hyperparameters, and we shall recommend a maximization over a small set of grid points as default. In addition, we investigate the performances of WARP at various choices of $\gamma$ in $\mathbb{B}_0$ including Laplace and quasi-Cauchy priors. We find out these $\mathbb{B}_0$ lead to almost exactly the same MSEs as normal priors (results not shown here). 
	
	We further investigate the sensitivity of WARP by considering the following ways to select hyperparameters $\tau$ and $\eta$: 
	\begin{itemize}
		\item $\tau$: ``function" (we use $\tau_j = 2^{-\alpha j} \tau_0$ as in~\ref{sec:experiments}); ``mix" (we use separate $\tau_j$ only for the last two levels and a constant for other levels, therefore we have three free parameters for $\tau$); ``full" (we use separate $\tau_j$'s for all levels $j$ ) 
		
		\item 
		$\eta$:  ``constant" (we use $\eta(A) = \eta_0$ for all $A$ as in~\ref{sec:experiments});``mix" (we use $\eta_j$ for the last two levels and a constant for other levels, therefore we have three free parameters for $\eta$); ``full" (we use separate $\eta_j$'s for all levels $j$). 
	\end{itemize} 
	\ref{table:more.sensitivity} shows that the MSEs only exhibit minimal differences across various combinations of tuning approaches. This confirms the previous findings that the proposed framework is not sensitive to hyperparameters.

	\begin{table*}[!t]
		\caption{\label{table:sensitivity.3D}Average MSEs ($\times 10^{-2}$) of WARP-full and WARP based on 100 replications under the setting of~\ref{table:3D}. 
		}
		\centering
		\vspace*{-1em}
		\begin{tabular}{lcccccccc}
			\toprule
			& \multicolumn{4}{c}{$n$ = 64} & \multicolumn{4}{c}{$n$ = 128}\\ \cmidrule(lr){2-5} \cmidrule(lr){6-9}
			Method & \multicolumn{2}{c}{$f = f_1$}& \multicolumn{2}{c}{$f = f_2$}  & \multicolumn{2}{c}{$f = f_1$}& \multicolumn{2}{c}{$f = f_2$} \\ \cmidrule(lr){1-1}\cmidrule(lr){2-3} \cmidrule(lr){4-5} \cmidrule(lr){6-7} \cmidrule(lr){8-9}
			& $\sigma = 0.1$ & $ 0.2$ & $\sigma = 0.1$ & $ 0.2$  &$\sigma = 0.1$ & $ 0.2$ & $\sigma = 0.1$ & $ 0.2$ \\
			{WARP-full}&{0.02}&{0.04}&{0.04}&0.12&{0.01}&{0.02}&{0.02}&{0.05}\\
			{WARP}&{0.02}&{0.04}&{0.04}&{0.11}&{0.01}&{0.02}&{0.02}&{0.05}\\
			\bottomrule
		\end{tabular}
	\end{table*} 
	
	\begin{table}[!t]
		\caption{\label{table:more.sensitivity}Sensitivity analysis of WARP when hyperparameters are selected differently using the Shepp-Logan phantom test image $(256\times 256)$ in Matlab at various $\sigma$. The average MSEs ($\times 10^{-2})$ are reported based on 5 replications.}
		\centering
		\begin{tabular}{llHHHHHHHlHlHlHl}
			\toprule 
			$\tau$ & $\eta$ & \multicolumn{7}{c}{}    &                                                                             
			0.1 & 0.2 & 0.3 & 0.4 & 0.5 & 0.6 & 0.7 \\  \cmidrule(lr){1-1} \cmidrule(lr){2-2}  \cmidrule(lr){9-16}
			function & constant &    0.07 &    0.13 &    0.18 &    0.24 &    0.28 &    0.34 &    0.37 &    0.03 &    0.13 &    0.27 &    0.42 &    0.57 &    0.72 &    0.89 \\ 
			function & mix &    0.07 &    0.13 &    0.18 &    0.24 &    0.28 &    0.33 &    0.37 &    0.03 &    0.13 &    0.27 &    0.42 &    0.58 &    0.73 &    0.88 \\ 
			function & full &    0.07 &    0.13 &    0.19 &    0.24 &    0.28 &    0.33 &    0.37 &    0.03 &    0.13 &    0.27 &    0.43 &    0.57 &    0.72 &    0.87 \\ 
			mix & constant &    0.07 &    0.14 &    0.20 &    0.25 &    0.30 &    0.35 &    0.40 &    0.03 &    0.13 &    0.26 &    0.41 &    0.57 &    0.72 &    0.94 \\ 
			mix & mix &    0.07 &    0.14 &    0.20 &    0.25 &    0.30 &    0.35 &    0.40 &    0.03 &    0.13 &    0.26 &    0.43 &    0.57 &    0.72 &    0.91 \\ 
			mix & full &    0.07 &    0.14 &    0.20 &    0.25 &    0.30 &    0.34 &    0.40 &    0.03 &    0.12 &    0.27 &    0.42 &    0.57 &    0.75 &    0.91 \\ 
			full & constant &    0.07 &    0.13 &    0.18 &    0.24 &    0.28 &    0.33 &    0.38 &    0.03 &    0.13 &    0.27 &    0.42 &    0.58 &    0.72 &    0.86 \\ 
			full & mix &    0.07 &    0.13 &    0.18 &    0.23 &    0.29 &    0.34 &    0.38 &    0.03 &    0.12 &    0.27 &    0.43 &    0.56 &    0.72 &    0.87 \\ 
			full & full &    0.07 &    0.13 &    0.18 &    0.23 &    0.28 &    0.33 &    0.37 &    0.03 &    0.13 &    0.27 &    0.42 &    0.57 &    0.72 &    0.88 \\ 
			\bottomrule
		\end{tabular}
	\end{table}

	\section{12 widely used test images}
	The 12 widely used test images used in Section~\ref{sec:DnCNN} are ploted in \ref{fig:set12}. 
	\setlength{\tabcolsep}{1pt}
	\begin{figure}
		\centering
		\begin{tabular}{cccccc}
			\includegraphics[width = 0.15\linewidth]{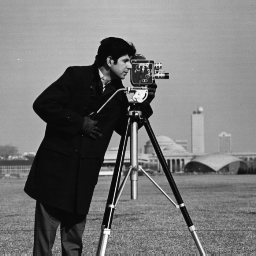}& \includegraphics[width=0.15\linewidth]{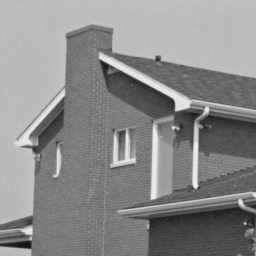} & \includegraphics[width=0.15\linewidth]{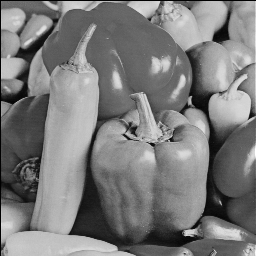} & \includegraphics[width=0.15\linewidth]{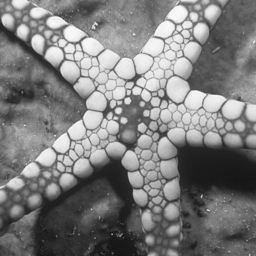} & \includegraphics[width=0.15\linewidth]{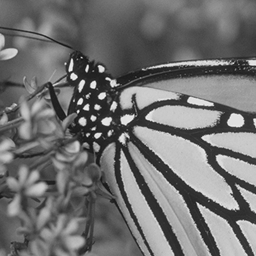} & \includegraphics[width=0.15\linewidth]{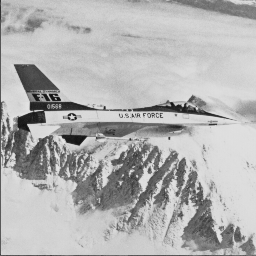}  \\
			\includegraphics[width = 0.15\linewidth]{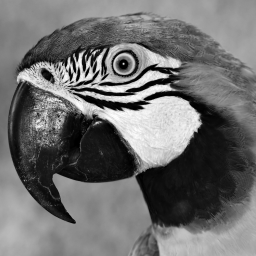}& \includegraphics[width=0.15\linewidth]{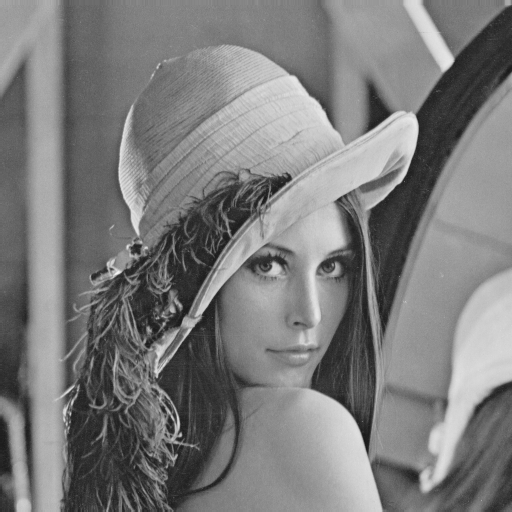} & \includegraphics[width=0.15\linewidth]{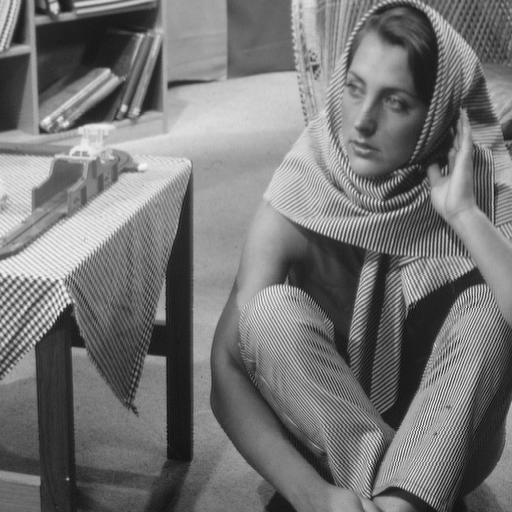} & \includegraphics[width=0.15\linewidth]{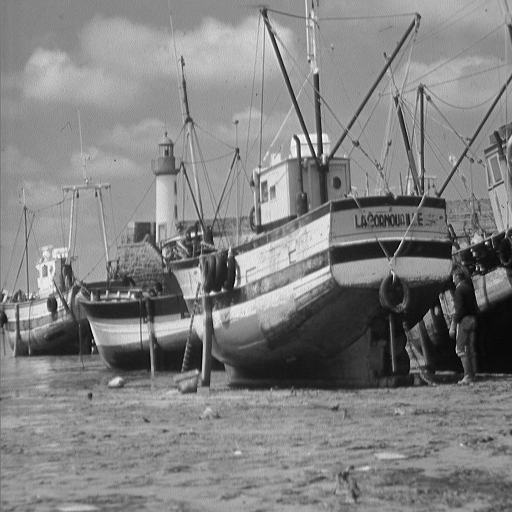} & \includegraphics[width=0.15\linewidth]{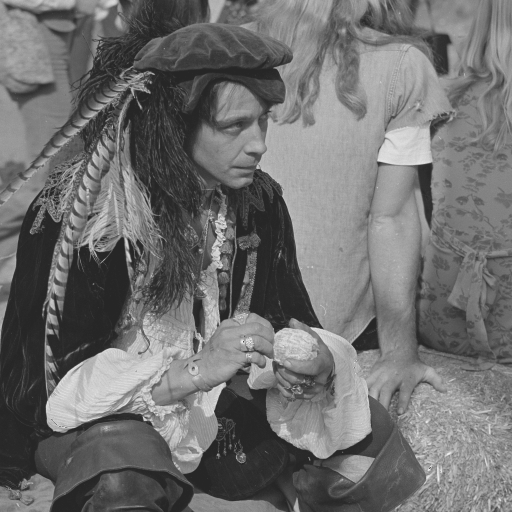} & \includegraphics[width=0.15\linewidth]{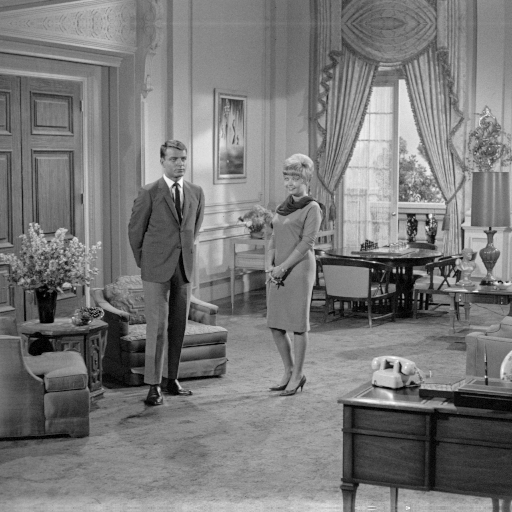}
		\end{tabular}
		\caption{The widely used 12 test images.}
		\label{fig:set12}
	\end{figure}

	\section{3D images\label{sec:experiment.3D} }
	Unlike WARP which is directly applicable to $m$-dimensional data for $m > 2$, other methods compared in~\ref{sec:2D} such as Wedgelet, TI-2D-Haar, and BPFA may require substantial modifications for a new dimensional setting. SHAH is conceptually applicable for 3D data, but the existing software takes hours to days in the tuning step for 3D images of intermediate size while its performance in 2D settings is not among top two. Therefore, we compare WARP with RM and a collection of other approaches, including a 3D image denoising method via local smoothing and nonparametric regression (LSNR) proposed by~\cite{Muk+Qiu:11}, anisotropic diffusion (AD) method~\citep{Per+Mal:90}, total variation minimization (TV) method~\citep{Rud+:92} and  optimized non-local means (ONLM) method~\citep{Cou+:08}. The TV method is modified by~\cite{Muk+Qiu:11} by minimizing a 3D-version of the TV criterion. We adopt simulation settings in~\cite{Muk+Qiu:11}, which uses two artificial 3D images with the following true intensity functions: 
	\begin{align}
	& f_1(x,y,z) =  -(x - \frac{1}{2})^2 -(y - \frac{1}{2})^2 -(z - \frac{1}{2})^2 \\ & \hspace{7em} + \I_{\{(x,y,z) \in R_1\cup R_2\}},
	\end{align} where $R_1 = \{(x,y,z): |x - \frac{1}{2}| \leq \frac{1}{4}, \; |y - \frac{1}{2}| \leq \frac{1}{4}, \; |z - \frac{1}{2}| \leq \frac{1}{4}\}$ and $R_2 = \{(x,y,z): (x - \frac{1}{2})^2 + (y - \frac{1}{2})^2 \leq 0.15^2, \;|z - \frac{1}{2}| \leq 0.35 \}$;
	\begin{align}
	& f_2(x,y,z) = \frac{1}{4}\sin(2\pi(x+y+z)+1) + \frac{1}{4} \\ & \hspace{5em} + \I_{ \{(x,y,z) \in S_1\cup S_2\}},
	\end{align} where $S_1 = \{(x,y,z): (x - \frac{1}{2})^2 + (y - \frac{1}{2})^2 \leq \frac{1}{4}(z - \frac{1}{2})^2 ,\;
	0.2 \leq z \leq 0.5 \}$ and $S_2 = \{(x,y,z): 0.2^2 \leq (x - \frac{1}{2})^2 + (y - \frac{1}{2})^2 + (z - \frac{1}{2})^2 \leq 0.4^2 ,\; z < 0.45 \}$.
	 
	\ref{table:3D} shows the comparison of various methods using MSE. It is worth mentioning that the numerical records for the other five methods to estimate $f_1$ and $f_2$ are from~\cite{Muk+Qiu:11} as the code is not immediately available and the running time for some method such as LSNR can take hours to days (including the tuning step). WARP is uniformly the best approach among all the selected methods at least under the simulation setting. 

	\begin{table*}[b]
		\caption{\label{table:3D}3D denoising for two images $f_1$, $f_2$ in terms of MSE $(\times10^{-2})$.
			WARP uses $5\times5\times5$ local shifts and are based on 100 replications.
			The mean of 100 MSEs is reported, and the maximum standard error is 0.00.
		} 
		\centering
		\begin{tabular}{lcccccccc}
			\toprule
			\multirow{3}{*}{Method} & \multicolumn{4}{c}{$n$ = 64} & \multicolumn{4}{c}{$n$ = 128}\\  \cmidrule(lr){2-5} \cmidrule(lr){6-9}
			& \multicolumn{2}{c}{$f = f_1$}& \multicolumn{2}{c}{$f = f_2$}  & \multicolumn{2}{c}{$f = f_1$}& \multicolumn{2}{c}{$f = f_2$} \\ \cmidrule(lr){2-3} \cmidrule(lr){4-5} \cmidrule(lr){6-7} \cmidrule(lr){8-9}
			& $\sigma = 0.1$ & $ 0.2$ & $\sigma = 0.1$ & $ 0.2$  &$\sigma = 0.1$ & $ 0.2$ & $\sigma = 0.1$ & $ 0.2$ \\ \cmidrule(lr){1-1} \cmidrule(lr){2-3} \cmidrule(lr){4-5} \cmidrule(lr){6-7} \cmidrule(lr){8-9}
			{WARP}&\textbf{0.02}&\textbf{0.04}&\textbf{0.04}&\textbf{0.11}&\textbf{0.01}&\textbf{0.02}&\textbf{0.02}&\textbf{0.05}\\ 
			LSNR    		& 0.03		& 0.08		& 0.06    & 0.13 & \textbf{0.01}		& \textbf{0.03}		& \textbf{0.02}    & 0.06		\\
			TV      		& 0.03		& 0.09		& 0.06    & 0.15 & \textbf{0.01}	& 0.04		& 0.03    & 0.06		\\
			AD      		& 0.06		& 0.35		& 0.07    & 0.38 & 0.03		& 0.20		& 0.04    & 0.22		\\
			ONLM    		& 0.03		& 0.12		& 0.06    & 0.14 & \textbf{0.01}		& 0.06		& 0.03    & 0.06		\\
			RM      		& 0.22		& 0.33		& 0.11    & 0.26 & 0.08		& 0.19		& 0.06    & 0.14		\\
			\bottomrule
		\end{tabular}
	\end{table*}

\end{NoHyper} 

\end{document}